\documentclass[11pt,leqno]{amsart}
\usepackage{amsfonts,amssymb,epsfig,latexsym}
\usepackage{amsmath,comment,bm,mathrsfs,braket,bbm}
\usepackage[latin1]{inputenc}
\usepackage{color,layout}
\usepackage{ifthen}
\usepackage{a4wide}
\usepackage{hyperref}
\hypersetup{%
  colorlinks = true,
  linkcolor  = blue,
  citecolor    = red
}

\usepackage{xparse}
\usepackage{graphicx}[dvipdfmx,hiresbb]
\usepackage{color}  
\usepackage{bmpsize} 

\newtheorem{theorem}{Theorem}
\newtheorem{lemma}{Lemma}[section]
\newtheorem{proposition}[lemma]{Proposition}

\newtheorem{remark}[lemma]{Remark}
\newtheorem{corollary}[lemma]{Corollary}

\newtheorem{condition}{Condition}

\def\square{\hbox{\vrule\vbox{\hrule\phantom{o}\hrule}\vrule}}

\graphicspath{{./graphics/}}

\newcommand{\be}{\begin{equation}}
\newcommand{\ee}{\end{equation}}
\newcommand{\ben}{\begin{equation*}}
\newcommand{\een}{\end{equation*}}

\newcommand{\til}[1]{\widetilde{#1}}

\numberwithin{equation}{section}

\newcommand{\N}{\mathbb{N}}
\newcommand{\Z}{\mathbb{Z}}
\newcommand{\R}{\mathbb{R}}
\newcommand{\C}{\mathbb{C}}

\newcommand{\cH}{{\mathcal H}}

\newcommand{\e}{\varepsilon}
\newcommand{\dl}{\delta}
\newcommand{\pphi}{\varphi}

\newcommand{\re}{{\rm Re}\hskip 1pt }
\newcommand{\im}{{\rm Im}\hskip 1pt }
\newcommand{\ord}{{\mathcal O}}

\newcommand{\ope}[1]{{\operatorname{#1}}}
\newcommand{\mc}[1]{{\mathcal{#1}}}

\newcommand{\qtext}[1]{\quad\text{#1 }\ }

\newcommand{\out}{{\sharp}}
\newcommand{\inc}{{\flat}}
\newcommand{\dir}{{\bullet}}
\newcommand{\supp}{{\ope{supp}\,}}

\numberwithin{equation}{section}

\begin{document}

\makeatletter
\@namedef{subjclassname@2020}{%
  2020 Mathematics Subject Classification}
\makeatother


\title[Resonant tunneling for quantum walks]{Resonant tunneling effect for quantum walks on directed graphs}
\author[Kenta Higuchi]{Kenta Higuchi}
\address[Kenta Higuchi]{Faculty of Education, Gifu University/ 
1-1 Yanagido, Gifu, 501-1193, Japan}
\email{higuchi.kenta.b7@f.gifu-u.ac.jp}


\maketitle

\begin{abstract}
Quantum walks that depend smoothly on a small parameter $\varepsilon\ge0$ are considered on directed graphs.
The asymptotic behavior of the scattering matrix of the quantum walk as $\varepsilon\to+0$ is investigated.
It is shown that, in this limit, the scattering matrix does not converge to that for $\varepsilon=0$ at points in the essential spectrum (the unit circle) that are asymptotically approached by a quantum resonance.
Furthermore, a phenomenon resembling and extending the resonant tunneling effect is observed by analyzing this discrepancy through resonant states.
\end{abstract}

\section{Introduction}
\subsection{A question concerning the resonant tunneling effect}
The resonant tunneling effect is one of the hallmark phenomena in quantum mechanics. 
It is a phenomenon in which the probability of a quantum particle with resonant energy (or frequency) tunneling through two potential barriers becomes greater than the probability of tunneling through a single potential barrier under some symmetry condition. 
Furthermore, the probability for double barrier attains almost one despite a vanishingly small probability for single barrier.
This phenomenon arises from the constructive interference of the quantum wave functions in the region between the two barriers.

Historically, the resonant tunneling effect was first suggested by a continuation of WKB solutions of the one-dimensional Schr\"odinger equation \cite{Bo}. Then it was experimentally demonstrated in \cite{CET}(see also \cite{Io,TE}), and has since been widely exploited in advanced electronic devices. 
The resonant tunneling effect in the double barrier setting is considered also for the Schr\"odinger operators on a two-dimensional unbounded domain in the context of waveguides \cite{BNPS} and the discrete-time quantum walk \cite{MMOS}.

The purpose of this manuscript is to answer the following question:\\
{\underline{\textbf{Question}} What are the underlying factors responsible for the resonant tunneling effect?}\\
It is known for the double barrier problem of one-dimensional Schr\"odinger operators that the resonant tunneling energy is close to a resonance.
Conversely, it is also known that the existence of resonance near the real line is not sufficient for observing the resonant tunneling effect.
In the double barrier problem, there exists a resonance exponentially (with respect to the semiclassical parameter derived from the Planck constant) close to each Dirichlet eigenvalue of the well between two potential barriers.
The transmission probability attains its local maxima near each Dirichlet eigenvalue.
However, the local maxima is still exponentially small when the ``size\footnote{The size of each potential barrier is expressed by the Agmon distance.}" of two barriers are different.
Therefore, a symmetry is required for the resonant tunneling effect in the double barrier problem.

A similar fact is also known for the double barrier problem for the quantum walks on the line \cite{MMOS, KKMS}.
We will review and improve these results in Section~\ref{sec:line}.
We will see that a symmetry for two barriers is required to observe the resonant tunneling effect also for this setting.
With these examples, one may naturally be led to the hypothesis: 
\textit{Some kind of symmetry of the equation is necessary to observe the resonant tunneling effect.}
However, we will also see that, in the triple barrier problem, the resonant tunneling effect can be observed \textbf{even if the barriers are not put symmetrically}.
This negotiates the above hypothesis.
In this manuscript, we show an \textbf{equivalent condition for observing the resonant tunneling effect} for a class of quantum walks (Corollary~\ref{cor:res-tun}).
The condition is not on a symmetry of the equation but on that of the \textbf{resonant state}.
Moreover, this applies to a larger class of quantum walks on directed graphs.

\subsection{Brief introduction of the setting and results}
We study quantum walks on directed graphs. 
The precise definition will be given later in Section~\ref{Sec:setting-results}. We provide a brief explanation here.
For a directed graph $(V,A)$, a quantum walk is defined by a unitary operator on the Hilbert space $\cH = \ell^2(A, \mathbb{C})$.
We impose on the unitary operator $U$ that, for any $\psi \in \cH$, the value of $U\psi$ at a directed edge $a \in A$ depends only on the values of $\psi$ at directed edges whose terminal vertex coincides with the initial vertex of $a$ (due to this condition, the directed graph $(V,A)$ must be \textit{balanced} in the sense of Subsection~\ref{sec:def-qw}). 
We assume that the directed graph $(V,A)$ is obtained by adding at most finitely many vertices and directed edges as a perturbation to a ``free graph." 
Here, a ``free graph" consists of $N\in \mathbb{N} \setminus \{0\}$ copies of the directed graph $(\mathbb{Z}, \{(j, j+1);\, j \in \mathbb{Z} \})$. 
For each copy of $(\mathbb{Z}, \{(j, j+1);\, j \in \mathbb{Z} \})$ of the free graph, the segment extending from $-\infty$ to the first intersection with the perturbed part is called the \textbf{incoming tail}, and the segment from the last intersection with the perturbed part to $+\infty$ is called the \textbf{outgoing tail}.

The unitary operator $U(\varepsilon)$ defining the quantum walk is assumed to \textbf{depend smoothly on $\varepsilon$} (Condition~\ref{cond:pert}), and furthermore, for $\varepsilon = 0$, the action of $U(0)$ on the free graph and the perturbed part are assumed to be completely independent \eqref{eq:free-motion-e=0}. 
That is, any function supported on the free graph remains supported only on the free graph after the application of $U(0)$, and vice versa.
Reflecting this dynamics, the spectrum of $U(0)$ restricted onto the free graph consists only of absolutely continuous spectrum $\mathbb{S}^1$ while that onto the perturbed part consists of discrete eigenvalues of modulus one.
We assume that each one of thes eigenvalues is simple (Condition~\ref{cond:pert}).
The dynamics also implies that the scattering matrix $\Sigma(0)=\Sigma(0,z)$ for $U(0)$ is diagonal \eqref{eq:scatt-e=0}.

This study focuses on investigating the \textbf{discrepancy} between this diagonal scattering matrix $\Sigma(0)$ and the scattering matrix $\Sigma(\e)=\Sigma(\e,z)$ $(z\in\mathbb{S}^1=\{z\in\C;\,|z|=1\})$ for $U(\varepsilon)$ in the limit $\varepsilon \to +0$:
\ben
\Sigma(\e,z)-\Sigma(0,z).
\een
Here, the scattering matrix $\Sigma(\e,z)$ is an $N \times N$ unitary matrix, which maps the vector of probability amplitudes of a generalized eigenfunction $\pphi\in\ell^\infty(A;\C)$ with $(U(\e)-z)\pphi=0$ on the incoming tails into that on the outgoing tails.
Note that, $z$ denotes the oscillation of each generalized eigenfunction in the time-evolution: $U(\e)^t\pphi=z^t\pphi=e^{it(\ope{arg} z)}\pphi$ $(t\in\Z)$.
Since the operator $U(\varepsilon)$ depends smoothly on the parameter $\varepsilon$, it may seem natural to expect that the scattering matrix $\Sigma(\e)$ for $U(\varepsilon)$ asymptotically approaches that of $U(0)$ (i.e. $\Sigma(0)$) as $\varepsilon \to +0$. 
However, it turns out that there exist exceptional oscillations $z$, those near \textbf{resonances} inside the unit circle $\mathbb{S}^1$, where such asymptotics fail. 

A resonance is a complex number $\lambda\in\C$ such that there exists a solution $\pphi$ (which is referred to as a \textbf{resonant state}) to $(U(\e)-\lambda)\pphi=0$ vanishing on the incoming tails (this definition includes the eigenvalues on the unit circle in the resonances).
In this setting, there exists a resonance of $U(\e)$ in $\ord(\e)$-neighborhood of each eigenvalue of $U(0)$ (see \cite{Kato} and Proposition~\ref{prop:matrix-reduction}).

It is also found that the discrepancy $\Sigma(\e,z)-\Sigma(0,z)$ arising in this case can be described in terms of the closest resonance to $z$ and associated (incoming and outgoing) resonant states  (Theorem~\ref{thm:non-resonant}). 
Furthermore, the operator norm of $\Sigma(\e,z)-\Sigma(0,z)$ is found to be almost $2$ at the closest oscillation $z\in\mathbb{S}^1$ to each resonance \eqref{eq:esti-M-below}.
Note that the value $2$ is the maximum as the operator norm of the difference between two unitary operators.

As a corollary of the results described above, it has also been revealed that a phenomenon 
which can be regarded as \textbf{a generalization of the resonant tunneling effect} occurs at oscillations close to resonances (Corollary~\ref{cor:res-tun}). 
For $\varepsilon = 0$, a wave incident on an incoming tail is observed with the same amplitude in the corresponding outgoing tail since $\Sigma(0)$ is diagonal. 
However, in the limit $\varepsilon \to +0$, this relationship breaks down significantly near each resonance inside the unit circle under a symmetry condition on the resonant state. 
More specifically, the scattered wave observed in the outgoing tails corresponding to the incoming tails where the incident wave is supported \textbf{asymptotically vanishes}, while a \textbf{definite amount of scattered wave appears} in the outgoing tails corresponding to the incoming tails where the incident wave is absent. 
In particular, when $N=2$, the scattering matrix is asymptotic to an anti-diagonal matrix.
This is nothing but the resonant tunneling effect.

\subsection{Related works on quantum walks}
The main subjects, the scattering matrix and resonances, of this study are mathematical terms of scattering theory.
The analysis of quantum walks using scattering theory originates from the work of Feldman and Hillery \cite{FeHi1}\footnote{There is an earlier study employing scattering theory for continuous-time quantum walks on trees also exists \cite{FaGu}
There is also a study of an abstract scattering theory for a class of unitary operators including certain class of quantum walks \cite{KaKu}}.
Their setting is similar to that in our study, that is, they introduced a quantum walk on balanced directed graphs with tails.
They pointed out motivations for using a scattering theoretical approach in the context of quantum algorithms.
Since then, extensive research has been conducted on scattering theory for quantum walks in different settings (e.g., \cite{MMOS,RST,Su}).

Several studies have focused on the analysis of scattering matrices, including \cite{HiSe,KKMS,Mo,MMOS}. 
For quantum walks on the line, the phenomenon of resonant tunneling in the double barrier case is investigated in \cite{MMOS}. 
In \cite{KKMS}, it was demonstrated that the $(j,k)$-entry of the $2\times2$ scattering matrix of a quantum walk on the line can be described by summing the probability amplitudes corresponding to every path incoming from $(-1)^k\infty$ and outgoing to $(-1)^{j-1}\infty$. 
The resonant tunneling effect in the double barrier case was revisited there. 
They however did not mention resonances of quantum walks.

Behavior of the scattering matrix for quantum walks on undirected graphs with tails, especially the Grover walk, has been studied intensively in recent years \cite{HSS,HiSe}.
The unitary operator of the Grover walk is automatically determined from the undirected graph on which it acts,  and the results reflect geometric features of the graph.
Remark that our setting includes the Grover walks as a particular case.
Furthermore, the ``comfortability" of quantum walks, which is the magnitude of the generalized eigenfunctions on the perturbed part, has also been studied for Grover walks \cite{AKMS,HSS}. 
This quantity is also related to the resonances (see Subsection~\ref{sec:comf}).

On the other hand, despite the significant role that  resonances play in scattering theory (see, e.g., the textbook \cite{DyZw} and the survey \cite{Zw17}), they have been introduced to quantum walks  relatively recently \cite{HMS}. 
In that study, a one-dimensional discrete-time quantum walk with a finite rank perturbation is considered.
Resonances were defined as poles of the analytically continued resolvent operator, and a resonance expansion of the time evolution was established. 
Then in \cite{HiMo}, the complex translation method was employed to deform the essential spectrum of the operator, allowing resonances to be analyzed as eigenvalues of the deformed operator. 
This approach was used to demonstrate how eigenvalues of an unperturbed quantum walk transform into resonances in a semi-bound state model of a two-dimensional quantum walk.

\subsection{Scattering theory and the resonant tunneling effect in other context}
Mathematically, the resonant tunneling effect can be formulated in terms of the scattering matrix and resonances.  
Both of them are key topics in the scattering theory and have been extensively studied in various settings.
Furthermore, the resonant tunneling effect is clear when a resonance approaches the essential spectrum.
It motivates us to consider asymptotic problem.
Here, we introduce some works on the scattering theory and asymptotic analysis in the contexts of Schr\"odinger operators and quantum graphs.

\subsubsection{Semiclassical Schr\"odinger operators}
Semiclassical analysis is one of the primary approaches for investigating Schr\"odinger operators. 
In this framework, the behavior of the operator is studied in the semiclassical limit $h\to+0$ by taking the Planck constant (or its variant obtained by dividing by the square root of the mass) as a small positive parameter $h$. 
As stated in Bohr's correspondence principle, classical physical quantities emerge in this limit.

The tunneling effect is a typical phenomenon of quantum mechanics, in which a particle penetrates a potential barrier. 
The transmission probability is computed via the Schr\"odinger equation. 
It is known to decay exponentially as $h\to+0$. 
However, in one-dimensional problems with two potential barriers, quantum resonances appear near the essential spectrum (the real axis) reflecting metastable states between the barriers. 
This is a type of shape resonance, and their influence gives rise to the resonant tunneling effect \cite{Bo}.

Nakamura \cite{Na1,Na2} analyzed the asymptotic behavior of the scattering matrix in the semiclassical limit for a multidimensional model of shape resonances.
He compared it with the scattering matrix of a modified model\footnote{He eliminated the metastable states by putting the Dirichlet condition on a sphere included in classically forbidden region.} from which the shape resonances had been removed. 
According to his study, the difference between the two scattering matrices is exponentially small at energies (points on the real axis) where no resonance is present \cite{Na1}, while near a resonance, the difference is primarily governed by the resonance and associated incoming and outgoing resonant states \cite{Na2}. 
At this point, Theorem~\ref{thm:non-resonant} in our manuscript shares similarities with these results, although the term ``resonant tunneling effect" is not explicitly used in its statement.


\subsubsection{Semiclassical matrix Schr\"odinger operators}
Schr\"odinger operators with Hermitian matrix-valued potentials, known as matrix Schr\"odinger operators, have attracted considerable attention in quantum chemistry. 
The analysis of matrix Schr\"odinger operators requires the simultaneous treatment of multiple classical dynamical systems, where each one corresponds to one of the eigenvalues of the matrix potential.
In \cite{AFH1}, a directed graph was introduced to study the asymptotic distribution of resonances for one-dimensional matrix Schr\"odinger operators: vertices represent the intersections of classical trajectories from different Hamiltonian systems, and directed edges correspond to each connected component of the classical trajectories divided by these intersections. 
Similar approaches have also been applied in investigations of eigenvalue splittings \cite{AsFu} and scattering matrices \cite{ABA}.

Furthermore, a model exhibiting resonant transmission and reflection near quantum resonances was reported by the author \cite{Hi}. 
This manuscript begins by studying quantum walks, as a first step toward explaining and generalizing the model, 
since quantum walks provide a simpler structure. 
A discrete model on a directed graph, analogous to those used in the studies mentioned above, is then employed to describe and extend the phenomenon.

\subsubsection{Quantum graphs}
A well-established correspondence exists between quantum walks on undirected graphs and quantum graphs \cite{HKSS}. 
Scattering theory for quantum graphs has been studied earlier than that for quantum walks (see e.g., the book \cite{BeKu}). 
In the scattering theory of quantum walks, one considers infinite sets of edges and vertices, referred to as  ``tails," whereas in quantum graphs, the analogous structure consists of semi-infinite edges known as ``leads." 
Furthermore, the phenomenon in which embedded eigenvalues transition into quantum resonances under perturbations has also been investigated \cite{ExLi1}. 
Although we could not find studies specifically on resonant tunneling, some researches have been conducted on the relationship between quantum resonances and the poles of the scattering matrix, as well as on the asymptotic distribution of resonances (see \cite{Lip} and references therein).

\subsubsection{Waveguides}
Waveguides have attracted attention as physical systems for studying scattering theory, and several mathematical models have been proposed, including the following:  

\noindent
1. \textit{Schr\"odinger operators with the Dirichlet boundary condition imposed on the side walls of the waveguide.} 
In this case, by the separation of variables into the transversal and longitudinal directions, the analysis of the operator is reduced to the one-dimensional problem with respect to the longitudinal direction.
Since the Dirichlet eigenvalues in the transversal direction increase as the cross-section narrows, narrow parts behave like potential barriers  (see the textbook \cite{BNPS} and references therein).
The resonant tunneling effect is investigated for this model. 
They regard the minimum cross-sectional area at the junction between the infinite and bounded regions as a parameter $\varepsilon$. 
When $\varepsilon = 0$, the system is completely disconnected, and the scattering matrix becomes a diagonal matrix.  
However, in the limit as $\varepsilon \to +0$, certain energy levels emerge where the diagonal components of the scattering matrix converge to zero due to the resonant tunneling effect.  

\noindent
2. \textit{Schr\"odinger operators on quantum graphs that incorporate the curvature of waveguide.}
This model arises in the limit where the cross-sectional area of the waveguide in the above model 
approaches zero (see \cite{DuEx}, \cite[Section 7.5]{BeKu}).  

\noindent
3. \textit{Schr\"odinger operators on manifolds with asymptotically cylindrical ends.} 
In this setting, to define the scattering matrix is already difficult.
There are many studies related to this problem (e.g., \cite{Ch,ChDa,IKL}).

\subsection{Method and construction of this manuscript}
The setting and main results are introduced in Section~\ref{Sec:setting-results}. 
In Subsection~\ref{sec:def-qw}, we make precise the definition of quantum walks on balanced directed graphs with tails.
In Subsection~\ref{sec:def-res-scatt} we introduce the resonances and the scattering matrix with some basic properties which will be proven in Subsection~\ref{sec:reduce-to-matrix}.
The main theorem and its corollary are stated in Subsection~\ref{sec:results}.
Theorem~\ref{thm:non-resonant} states that the discrepancy $\Sigma(\e,z)-\Sigma(0,z)$ is of finite rank and written in terms of resonances and associated outgoing and incoming resonant states. 
It also shows that this discrepancy is small of $\ord(\e)$ when $z$ is not an eigenvalue of the non-perturbed operator $U(0)$, while its operator norm attains almost two near each resonance inside the unit circle.
Corollary~\ref{cor:res-tun} gives the condition on the resonant state to occur a generalized version of the resonant tunneling effect.
We also provide a result (Theorem~\ref{thm:comfortability}) on a relation between the comfortability and resonances in Subsection~\ref{sec:comf}.

Section~\ref{sec:line} is devoted to simple examples of quantum walks on the line.
Computations based on the transfer matrices are omitted there, and written in Appendix~\ref{app:on-the-line}.
The double barrier problem is revisited in Subsection~\ref{sec:DB}.
We see that a slight change of the symmetry of two barriers may completely break the resonant tunneling effect.
The triple barrier problem is considered in Subsection~\ref{sec:TB}.
We give an example where the resonant tunneling effect takes place with three non-symmetric  barriers.

In Section~\ref{sec:examples}, two concrete examples of resonant tunneling are given.
One given in Subsection~\ref{sec:ex-Sch} is motivated by the study of matrix Schr\"odinger operators \cite{AFH1}, the other one given in Subsection~\ref{sec:model-cycle} is very simple but many types of the ``generalized" resonant tunneling effect stated in Corollary~\ref{cor:res-tun} take place on this model.
The resonant tunneling effect for each model is given by the formulas \eqref{eq:res-tun-ms}, \eqref{eq:res-tun-c1}, \eqref{eq:res-tun-c2}, \eqref{eq:res-tun-c3}, and \eqref{eq:res-tun-c4}, and is illustrated in Figures~\ref{Fig-res-tun-ms} and \ref{Fig-res-tun-cycle}.

In Section~\ref{sec:S-matrix-res}, an expression of the scattering matrix in terms of resonant states is given as Theorem~\ref{thm:res-exp}.
This expression plays an important role in the proof of Theorem~\ref{thm:non-resonant}.
In Subsection~\ref{sec:reduce-to-matrix}, we show some basic properties of resonances and the scattering matrix in terms of a finite rank, non-normal  operator concerning the perturbed part of the graph.
The proof given in Subsection~\ref{sec:prf-res-exp} of Theorem~\ref{thm:non-resonant} relies on the general (and elementary) theory of matrices.

Section~\ref{sec:prf-main} is devoted to the proof of the main results stated in Subsections~\ref{sec:results} and \ref{sec:comf}.
Theorem~\ref{thm:non-resonant} is proven by using the perturbation theory for matrices (see e.g., \cite{Kato}).
We also used the reduced resolvent operator which corresponds to the resolvent operator restricted onto the absolutely continuous subspace of $U(\e)$.
Corollary~\ref{cor:res-tun} follows essentially from the condition of the equality of the triangular inequality.
Remember that the operator norm of the discrepancy $\Sigma(\e,z)-\Sigma(0,z)$ reaches almost 2, which is the maximum as a difference between two unitary operators.
Theorem~\ref{thm:comfortability} is shown in a similar argument as that for Theorem~\ref{thm:non-resonant}.

\section{Setting and main results}\label{Sec:setting-results}
In this section, we state our main results, Theorem~\ref{thm:non-resonant} and Corollary~\ref{cor:res-tun}.
To provide a precise statement, we first define a quantum walk on balanced directed graphs with tails in Subsection~\ref{sec:def-qw}.
We then introduce the definition and basic properties of resonances and the scattering matrix of such a quantum walk in Subsection~\ref{sec:def-res-scatt}.
We state our main results on the asymptotic behavior of the scattering matrix in Subsection~\ref{sec:results}.
The dependence of the quantum walk $U(\e)$ on the small parameter $\e\ge0$ and a condition on the free quantum walk $U(0)$ are also made precise in this subsection.
A result concerning the comfortability is stated in Subsection~\ref{sec:comf}.

\subsection{Definition of quantum walks on balanced directed graphs with tails}\label{sec:def-qw}
\begin{figure}
\centering
\includegraphics[bb=0 0 266 179, width=7cm]{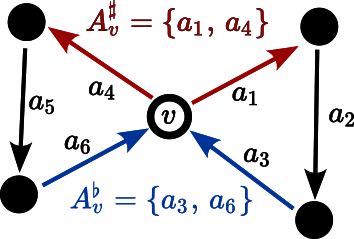}
\caption{An example of a balanced directed graph}
\label{Fig1}
\end{figure}
A directed graph is a pair of two sets $V$ and $A$ called \textit{vertex set} and \textit{arc set} equipped with the two maps $A\ni a\mapsto a^-\in V$ and $A\ni a\mapsto a^+\in V$.
For each arc $a\in A$, the vertices $a^-$ and $a^+$ are called \textit{origin} and \textit{terminus} of $a$.
To define a unitary time-evolution, we assume that the directed graph $(V,A)$ is \textit{balanced} (also called \textit{Eulerian}), that is, for each $v\in V$, the set $A_v^\inc:=\{a\in A;\,a^+=v\}$ of incoming arcs to $v$ and the set $A_v^\out:=\{a\in A;\,a^-=v\}$ of outgoing arcs from $v$ have the same (finite) number of elements (see Figure~\ref{Fig1}). 
We call this number \textit{degree} of $v$, and denote it by $\deg v$.
Then the arc set $A$ is decomposed into a disjoint union in two ways:
\be\label{eq:arc-decomp}
A=\bigsqcup_{v\in V}A_v^\inc=\bigsqcup_{v\in V}A_v^\out.
\ee

The time-evolution operator $U$ of quantum walk is defined on $\C^A$,
whose restriction onto the Hilbert space $\cH=\ell^2(A)\subset \C^A$ is unitary.
In the same way as the decomposition \eqref{eq:arc-decomp} of the arc set, the function space $\C^A$ is decomposed as 
\ben
\C^A=\bigoplus_{v\in V}\C^{A_v^\inc}=\bigoplus_{v\in V}\C^{A_v^\out}.
\een
The operator $U$ is defined for a family of isometric operators $\{U_v:\C^{A_v^\inc}\to \C^{A_v^\out};\,v\in V\}$ by
\ben
U=\sum_{v\in V}\chi(A_v^\out)^*U_v\chi(A_v^\inc),
\een
where $\chi(B):\C^A\to\C^{B}$ for $B\subset A$ is the restriction operator onto $B$, and its adjoint $\chi(B)^*$ is the extension by zeros:
\ben
\chi(B)\psi(a)=\psi(a)\quad(a\in B),\qquad
\chi(B)^*u(a)=\left\{
\begin{aligned}
&u(a)&&a\in B\\
&0&&\text{otherwise},
\end{aligned}\right.\quad(\psi\in\C^A,\  u\in\C^{B}).
\een
Remark that $\chi(B)\chi(B)^*:\C^B\to \C^B$ is the identity operator, and that $\chi(B)^*\chi(B):\C^A\to\C^A$ is the multiplication operator of the characteristic function $\mathbbm{1}_B:A\to\{0,1\}$ of the set $B$.
By definition, $U\psi(a)$ for each $a\in A$ is given by
\ben
U\psi(a)=\sum_{b\in A_v^\inc}(U_v)_{a,b}\psi(b),
\een
where we put $v=a^-$, and $(U_v)_{a,b}=(U_v\dl_b,\dl_a)$ is the $(a,b)$-entry of the $\deg v\times \deg v$-representation matrix  of $U_v$ with the bases $\{\dl_b;\,b\in A_v^\dir\}$ of $\C^{A_v^\dir}$ $(\dir=\inc,\out)$.
Here and after, $\dl_b$ denotes the characteristic function of $\{b\}$.
Note that the value at each arc of $U\psi$ depends only on the values of $\psi$ at a finite number of arcs.
Thus, the definition of $U$ is also valid for functions on $A$ which does not belong to $\cH$.

\begin{figure}
\centering
\includegraphics[bb=0 0 665 458, width=8cm]{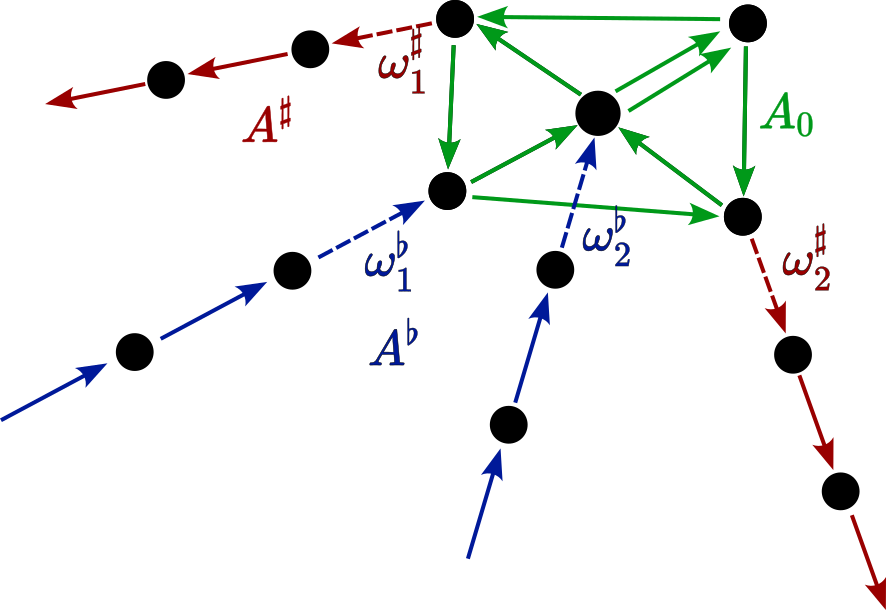}
\caption{An example of a balanced directed graph with tails}
\label{Fig2}
\end{figure}
The scattering theory is developed for operators with absolutely continuous spectrum.
Accordingly, we consider quantum walks on infinite graphs as follows.
A balanced directed graph $(V,A)$ is assumed to be decomposed into three parts:
\ben
V=V_0\cup V^\inc\cup V^\out,\quad
A=A_0\sqcup A^\inc\sqcup A^\out.
\een
Each one of $(V_0,A_0)$, $(V^\inc,A^\inc)$, $(V^\out,A^\out)$ forms a graph.
The graph $(V_0,A_0)$ is finite while the others are composed of $N\in\N\setminus\{0\}$-copies of simple, semi-infinite graphs called \textit{tails}.
More precisely, one has the decomposition
\begin{align*}
(V^\dir, A^\dir)=\bigcup_{n=1}^N(V_n^\dir, A_n^\dir)\quad(\dir=\inc,\out),
\end{align*}
where for each $n$, $(V_n^\inc,A_n^\inc)$ is called an \textit{incoming tail} while $(V_n^\out,A_n^\out)$ is called an \textit{outgoing tail}.
The vertex set and the arc set of each tail are given by 
\ben
V_n^\dir=\{v_{n,j}^\dir;\,j\in\N\},\quad
A_n^\inc=\{a_{n,j}^\inc;\,j\in\N\},\quad
A_n^\out=\{a_{n,j}^\out;\,j\in\N\setminus\{0\}\},
\een
with $a_{n,j}^\inc=(v_{n,j+1}^\inc,v_{n,j}^\inc)$ and $a_{n,j}^\out=(v_{n,j-1}^\out,v_{n,j}^\out)$ (in the notation $a=(a^-,a^+)$).
For each $(n,\dir)\in\{1,2,\ldots,N\}\times\{\inc,\out\}$, the vertex $v_{n,0}^\dir$ is identified with a vertex of $V_0$, and is the unique intersection: $V_n^\dir\cap V_0=\{v_{n,0}^\dir\}$.
Put $\omega_n^\inc:=a_{n,0}^\inc$, $\omega_n^\out:=a_{n,1}^\out$, and
\be
\Omega^\dir:=\{\omega_n^\dir;\,n=1,2,\ldots,N\}\quad(\dir=\inc,\out).
\ee

The dynamics induced by $U$ on the tails are trivial, that is, $\deg(v_{n,j}^\dir)=1$ for each $(n,j,\dir)\in\{1,2,\ldots,N\}\times (\N\setminus\{0\})\times\{\inc,\out\}$ implies that $U_{v_{n,j}^\dir}$ is just a composition of the multiplication by a constant $c_{n,j}^\dir$ of modulus one and the translation:
\ben
U\psi(a^\inc_{n,j})=c_{n,j}^\inc\psi(a^\inc_{n,j+1}),\qquad
U\psi(a^\out_{n,j})=c_{n,j}^\out\psi(a^\out_{n,j-1}).
\een
Note that, by the action of $U$, a state on incoming tails approaches $A_0$, whereas one on outgoing tails moves away from $A_0$.
This is why we call them incoming and outgoing.
For simplicity, throughout this manuscript, we assume $c_{n,j}^\inc=1$ for every $v\in V\setminus V_0$. 
It implies
\be\label{eq:dyn-on-tails}
U\psi(a^\inc_{n,j})=\psi(a^\inc_{n,j+1}),\qquad U\psi(a^\out_{n,j})=\psi(a^\out_{n,j-1}).
\ee

\begin{remark}\label{rem:make_tails}
We can make such a balanced directed graph $(V,A)$ by modifying a given finite balanced directed graph $(V_1,A_1)$ (see Figure~\ref{Fig-MT}).
Let us take a subset $\dl V_1\subset V_1$. 
We make $\dl V_1$ ``the set of vertices at infinity."
Put $V_0=V_1\setminus\dl V_1$, $A_0=A_1\setminus(\Omega^\inc\cup\Omega^\out)$ with $\Omega^\inc:=\bigcup_{v\in\dl V_1}(A_1)_v^\inc$, $\Omega^\out:=\bigcup_{v\in\dl V_1}(A_1)_v^\out$.
Then the graph $(V,A)$ is given by the union of $(V_0,A_0)$ and incoming and outgoing tails associated with each element of $\Omega^\inc$ and $\Omega^\out$.
\end{remark}
\begin{figure}
\centering
\includegraphics[bb=0 0 991 311, width=14cm]{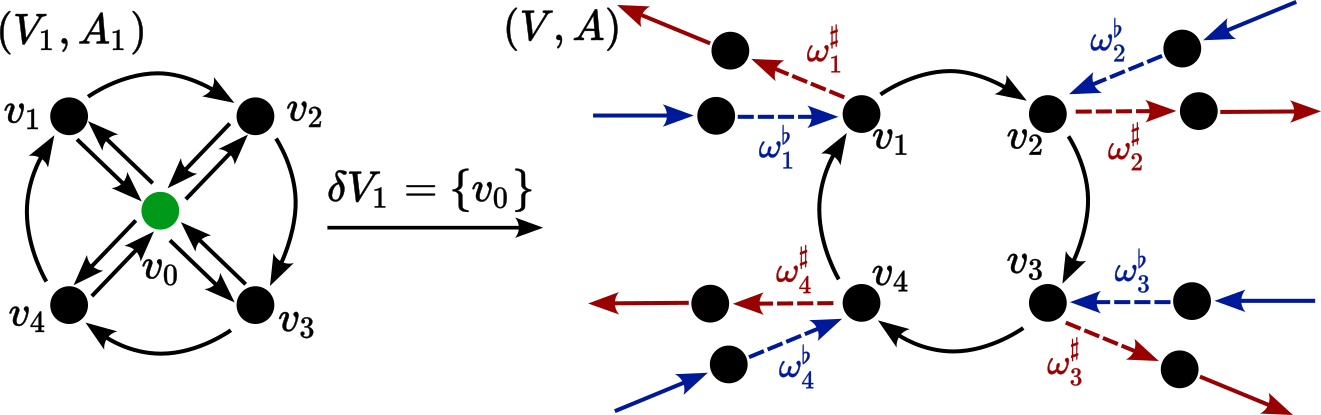}
\caption{Construction of $(V,A)$ from a finite graph $(V_1,A_1)$ as in Remark~\ref{rem:make_tails}}
\label{Fig-MT}
\end{figure}

\subsection{Resonances and scattering matrix}\label{sec:def-res-scatt}
We here introduce the scattering matrix and the resonances in a simple way.
\footnote{In the formal way, one introduces a ``free system," that is, another unitary operator on a Hilbert space.
The scattering matrix is defined from the scattering operator through a spectral representation of the free unitary operator. 
The scattering operator is given by the ``relative" time evolution of $U$ compared with that of the free unitary operator.
The free system for our setting should be defined on the graph consists of 
$N$-copies of infinite lines $(\Z,\{(j,j+1);\,j\in\Z\})$.
The negative and the positive half of each infinite lines may identified with an incoming and an outgoing tail.
The abstract scattering theory for unitary operator with two-Hilbert spaces is studied in \cite{Ti}.
Resonances are defined as the poles of the meromorphic continuation of the cut-off resolvent $\mathbbm{1}_B(U-z)^{-1}\mathbbm{1}_B$ with $B\subset A$ large enough beyond the essential spectrum $\mathbb{S}^1$ of $U$ \cite{HMS}.
}
We say that a function $\pphi\in\C^A$ is \textit{outgoing} if the intersection of the support $\ope{supp}\,\pphi=\{a\in A;\,\pphi(a)\neq0\}$ and the incoming tails $A^\inc$ is finite.
We say that $z\in\C$ is a \textit{resonance} of $U$ if there exists a non-trivial outgoing function $\varphi\in\C^A$ such that $(U-z)\varphi=0$. 
(It is also referred to as \textit{outgoing resonance} to clearly distinguish the difference with incoming resonance defined later.)
We call such an outgoing function a \textit{resonant state} associated with the resonance $z$.
We regard $z=0$ as a resonance of $U$ even the above definition does not apply to it.
The modulus of each resonance is smaller or equal to one.
Let $\ope{Res}(U)$ be the set of resonances of $U$.
The \textit{multiplicity} $m=m(z)$ of each non-zero resonance $z\in\ope{Res}(U)\setminus\{0\}$ is the dimension of the set of outgoing functions $\pphi$ such that $(U-z)^k\pphi=0$ holds for some $k\in\N$.
The multiplicity $m$ is bounded by $\#A_0<\infty$.
These properties of resonances will be shown as Proposition~\ref{prop:matrix-reduction}.

We also introduce incoming resonances.
A function $\varphi\in\C^A$ is said to be \textit{incoming} if $\ope{supp}\,\varphi\cap A^\out$ is a finite set.
We say that $z\in\C$ is an \textit{incoming resonance} of $U$ if there exists a non-trivial incoming function $\varphi\in\C^A$ such that $(U-z)\varphi=0$.
Such an incoming function $\pphi$ is called \textit{incoming resonant state} associated with the incoming resonance $z$.
Let $\ope{Res}^\inc(U)$ denote the set of incoming resonances.
A complex number $z\in\C\setminus\{0\}$ is an outgoing resonance if and only if $\bar{z}^{-1}$ is an incoming resonance.
The multiplicity of each resonance $z\in\ope{Res}(U)$ coincides with that of the incoming resonance $\bar{z}^{-1}$.
For each resonant state $\pphi$ associated with a resonance $z\in\ope{Res}(U)$ of multiplicity one, there uniquely exists an incoming resonant state $\pphi^\circledast$ associated with $\bar{z}^{-1}\in\ope{Res}^\inc(U)$ such that 
\be\label{eq:normalize-res-st}
(\pphi,\pphi^\circledast)=\sum_{a\in A}\pphi(a)\overline{\pphi^\circledast(a)}
=\sum_{a\in A_0}\pphi(a)\overline{\pphi^\circledast(a)}=1.
\ee
In fact, the support of each resonant state (resp. incoming resonant state) does not intersect with $A^\inc$ (resp. $A^\out$).
The following properties concerning incoming resonances are shown as Proposition~\ref{prop:matrix-reduction2}.
Moreover, it is shown as Lemma~\ref{eq:sym-out-inc-res-st} that they satisfy for any $n\in\{1,2,\ldots,N\}$,
\be\label{eq:dist-res-state}
\frac{|\pphi_{\lambda,\e}(\omega_n^\out)|}{\|\chi(\Omega^\out)\pphi_{\lambda,\e}\|}
=
\frac{|\pphi_{\lambda,\e}^\circledast(\omega_n^\inc)|}{\|\chi(\Omega^\inc)\pphi_{\lambda,\e}^\circledast\|}+\ord(\e).
\ee

We define the scattering matrix by using the generalized eigenfunctions.
For each $z\in\mathbb{S}^1=\{z\in\C;\,|z|=1\}$ and $\alpha^\inc\in\C^{\Omega^\inc}$, there exists a function $\varphi\in\C^A$ such that 
\be\label{eq:gen-ef}
(U-z)\varphi=0,\quad\text{and}\quad\chi(\Omega^\inc)\varphi=\alpha^\inc.
\ee 
It defines a function $\alpha^\out:=z\chi(\Omega^\out)\varphi\in\C^{\Omega^\out}$.
The function $\varphi$ is called a \textit{generalized eigenfunction}.
The function $\pphi$ is not unique if $z$ is an eigenvalue of $U$, however, $\alpha^\out$ is always unique for $z\in\mathbb{S}^1$ (Proposition~\ref{prop:gen-ef}).
Then the scattering matrix $\Sigma(z):\C^{\Omega^\inc}\to\C^{\Omega^\out}$ is defined as the isometric linear map $\alpha^\inc\mapsto\alpha^\out$.

\subsection{Asymptotic behavior of the scattering matrix in the semi-classical limit}\label{sec:results}
Let $\{U(\e);\,0\le \e\le \e_0\}$ $(\e_0>0)$ be a family of unitary time-evolution operators of a quantum walk on $\cH$.
Each operator $U(\e)$ is induced by local isometries $U_v(\e):\C^{A_v^\inc}\to\C^{A_v^\out}$.
We assume that the dynamics induced by $U(0)$ is deterministic from each incoming tails to an outgoing tail, that is, $\ope{supp}(U(0)^k\dl_a)$  is a singleton for any $a\in A^\inc\cup A^\out$ and $k\in\Z$.
Then one can label the arcs $\Omega^\inc=\{\omega_1^\inc,\omega_2^\inc,\ldots,\omega_N^\inc\}$ and $\Omega^\out=\{\omega_1^\inc,\omega_2^\inc,\ldots,\omega_N^\inc\}$ so that one has
\be\label{eq:free-motion-e=0}
U^{k_n}\dl_{\omega_n^\inc}=c_n\dl_{\omega_n^\out}\qquad(n=1,2,\ldots,N)
\ee
for some $k_n\in\N\setminus\{0\}$ and $c_n\in\mathbb{S}^1$.
It follows that the scattering matrix $\Sigma(0,z)$ of $U(0)$ is given by
\be\label{eq:scatt-e=0}
\Sigma(0,z)=z\,\ope{diag}(z^{-k_1}c_1,z^{-k_2}c_2,\ldots,z^{-k_N}c_N).
\ee

Remark that each eigenvalue of $U(\e)$ is semi-simple, that is, the algebraic and geometric multiplicities of each eigenvalue coincide with each other.
This is a consequence that $U(\e)$ is unitary, in particular normal.
We assume a stronger condition.
\begin{condition}\label{cond:pert}
Every eigenvalue of $U(0)$ is simple.
The operator-valued function $U_v(\e)$ for each $v\in V_0$ is right-differentiable at $\e=0$.
\end{condition}
Under this condition, for each $\lambda\in\ope{Res}(U(0))\cap\mathbb{S}^1$ and for $\e>0$ small enough (this means for any $\e\in(0,\e_1]$ with some $\e_1>0$), there exists $\lambda_\e=\lambda_\e(\lambda)\in\ope{Res}(U(\e))$ such that $|\lambda_\e-\lambda|=\ord(\e)$ (see \cite{Kato} and Proposition~\ref{prop:matrix-reduction}).
The algebraic multiplicity of $\lambda_\e$ is one for $\e>0$ small enough, and we can take a resonant state $\pphi_{\lambda,\e}$ associated with $\lambda_\e$ such that $\|\chi(A_0)(\pphi_{\lambda,\e}-\pphi_{\lambda,0})\|=\ord(\e)$.
We define the (rank at most one) operator-valued function $M_{\lambda,\e}=M_{\lambda,\e}(z):\C^{\Omega^\inc}\to\C^{\Omega^\out}$ by 
\be\label{eq:def-M_lam}
M_{\lambda,\e}(z)\alpha^\inc=
\frac{\lambda_\e^2}{z-\lambda_\e}(\alpha^\inc,\chi(\Omega^\inc)\pphi^\circledast_{\lambda,\e})_{\C^{\Omega^\inc}}\chi(\Omega^\out)\pphi_{\lambda,\e}
\qquad(\alpha^\inc\in\C^{\Omega^\inc}).
\ee
Remark that $z=\lambda_\e$ is a singular point of the function $z$ when $\lambda_\e\in\mathbb{S}^1$.
It is removable by defining $M_{\lambda,\e}(\lambda_\e)=0$ since $\chi(\Omega^\inc)\pphi_{\lambda,\e}^\circledast=0$ holds in that case (see Lemma~\ref{lem:on-tails} and Proposition~\ref{prop:matrix-reduction} (\ref{enu:red3})).

\begin{theorem}\label{thm:non-resonant}
Under Condition~\ref{cond:pert}, one has
\be\label{eq:discrepancy-thm}
\Sigma(\e,z)-\Sigma(0,z)=\sum_{\lambda\in\ope{Res}(U(0))\cap\mathbb{S}^1}M_{\lambda,\e}(z)+\ord_{\C^{\Omega^\inc}\to\C^{\Omega^\out}}(\e).
\ee
Moreover, for any fixed $z\in\mathbb{S}^1\setminus\ope{Res}(U(0))$ independent of $\e$, $M_{\lambda,\e}(z)=\ord(\e)$ holds for every $\lambda$, and
\ben
\left\|\Sigma(\e,z)-\Sigma(0,z)\right\|_{\C^{\Omega^\inc}\to\C^{\Omega^\out}}=\ord(\e).
\een
One also has an estimate from below: 
\ben
\|M_{\lambda,\e}(z)\|_{\C^{\Omega^\inc}\to\C^{\Omega^\out}}
=\left|\frac{1-|\lambda_\e|^2}{z-\lambda_\e}\right|\|\chi(A_0)\pphi_{\lambda,\e}\|\|\chi(A_0)\pphi^\circledast_{\lambda,\e}\|
\ge
\left|\frac{1-|\lambda_\e|^2}{z-\lambda_\e}\right|.
\een
This implies that there exists $C>0$ such that
\be\label{eq:esti-M-below}
\left\|M_{\lambda,\e}\left(\frac{\lambda_\e}{|\lambda_\e|}\right)\right\|_{\C^{\Omega^\inc}\to\C^{\Omega^\out}}
\ge1+|\lambda_\e|\ge2-C\e,
\ee
provided that $|\lambda_\e|<1$.
\end{theorem}

For a non-empty set $J\subsetneq\{1,2,\ldots,N\}$, we denote the complement by $J^c=\{1,2,\ldots,N\}\setminus J$, and put
\ben
\Omega_J^\dir=\{\omega_n^\inc\in\Omega^\dir;\,n\in J\}\qquad(\dir=\inc,\out).
\een
We define the reflection and transmission probabilities $R$ and $T$ for $\e\ge0$, $z\in\mathbb{S}^1$ and for each normalized incoming data $\alpha^\inc\in\C^{\Omega^\inc}$ $(|\alpha^\inc|=1)$ supported in $\Omega^\inc_J$ $(\ope{supp}\,\alpha^\inc\subset\Omega^\inc_J)$ by 
\be
R=R(J,\alpha^\inc,\e,z):=\left\|\chi(\Omega_J^\out)\Sigma(\e,z)\alpha^\inc\right\|^2,\quad
T=T(J,\alpha^\inc,\e,z):=\left\|\chi(\Omega_{J^c}^\out)\Sigma(\e,z)\alpha^\inc\right\|^2.
\ee
The unitarity of the scattering matrix implies $T+R=1$.
Note that for $\e=0$, $T(J,\alpha^\inc,0,z)=0$ holds for any $J$, $\alpha^\inc$, $z$.
As a consequence of Theorem~\ref{thm:non-resonant}, we see a generalized version of the resonant tunneling effect.
\begin{corollary}\label{cor:res-tun}
For any fixed $z\in\mathbb{S}^1\setminus\ope{Res}(U(0))$, any non-empty subset $J\subsetneq\{1,2,\ldots,N\}$, and any normalized $\alpha^\inc\in\C^{\Omega^\inc}$ supported in $\Omega_J^\inc$, one has 
\be\label{eq:non-resonant-transm}
T(J,\alpha^\inc,\e,z)=\ord(\e^2)\quad\text{as }\e\to+0.
\ee
Fix $\lambda\in\ope{Res}(U(0))\cap\mathbb{S}^1$ such that $|\lambda_\e(\lambda)|<1$ for positive small $\e>0$.
Assume that there exists $J\subset\{1,2,\ldots,N\}$ such that 
\be\label{eq:cond-res-tun}
\|\chi(\Omega_J^\inc)\pphi_{\lambda,\e}^\circledast\|
=\|\chi(\Omega_{J^c}^\inc)\pphi_{\lambda,\e}^\circledast\|+\ord(\e).
\ee
Then there exists a constant $c_{\e,J}$ of modulus one such that
\be\label{eq:resonant-tunneling-explicit}
\Sigma(\e,\lambda_\e/|\lambda_\e|)\alpha_J^\inc
=c_{J,\e}
\alpha_{J^c}^\out
+\ord(\sqrt{\e}),
\ee
with 
\be\label{eq:alpha-Js}
\alpha_J^\inc=\frac{\mathbbm{1}_{\Omega_J^\inc}\chi(\Omega^\inc)\pphi_{\lambda,\e}^\circledast}{\|\mathbbm{1}_{\Omega_J^\inc}\chi(\Omega^\inc)\pphi_{\lambda,\e}^\circledast\|},\qquad
\alpha_{J^c}^\out=\frac{\mathbbm{1}_{\Omega_{J^c}^\out}\chi(\Omega^\out)\pphi_{\lambda,\e}}{\|\mathbbm{1}_{\Omega_{J^c}^\out}\chi(\Omega^\out)\pphi_{\lambda,\e}\|}.
\ee
In particular, one has
\be\label{eq:supp-disj}
T(J,\alpha_J^\inc,\e,\lambda_\e/|\lambda_\e|)=1+\ord(\e).
\ee
In addition, 
there exist $\theta_-,\theta_+\in\R$ such that $\theta_\pm=\pm(1-|\lambda_\e|)+\ord((1-|\lambda_\e|)^2)$, and that
\be\label{eq:width-of-peak}
T(J,\alpha_J^\inc,\e,e^{i\theta}\lambda_\e/|\lambda_\e|)\ge\frac12,
\ee
for $\theta_-\le\theta
\le\theta_+$.
\end{corollary}

Note that $\lambda_\e/|\lambda_\e|$ is the closest point on the unit circle to the resonance $\lambda_\e$.
The well-known setting of the resonant tunneling corresponds to the case with $N=2$. 
Then the subset $J$ is forced to be a singleton, and the choice of normalized incoming data for this case is essentially unique.
We see from the above corollary that the condition for the resonant tunneling effect is a symmetry on the resonant state.
We will confirm this for simple examples in Section~\ref{sec:line}.


Note that the function $\Sigma(0,z)\alpha_J^\inc$ is supported only on $\Omega_J^\out$ for any $z\in\mathbb{S}^1$, and one in particular has
\ben
\left\|\chi(\Omega_J^\out)\Sigma(0,z)\alpha_J^\inc\right\|=1.
\een
This shows a drastic difference between $\Sigma(0,z)$ and the limit as $\e\to+0$ of $\Sigma(\e,z)$ despite $U(\e)\to U(0)$ as $\e\to+0$.
The difference $\theta_+-\theta_-$ is called the \textit{width of the resonant peak} at its half-height.
The above corollary stated that this width is approximated by twice the width of the resonance $\lambda_\e$, that is, the distance from the essential spectrum to the resonance.

\begin{remark}\label{rem:equiv-cond-res-tun}
The condition \eqref{eq:cond-res-tun} is equivalent to 
\be\label{eq:cond-res-tun1}
\|\chi(\Omega_J^\out)\pphi_{\lambda,\e}\|
=\|\chi(\Omega_{J^c}^\out)\pphi_{\lambda,\e}\|+\ord(\e),
\ee
due to \eqref{eq:dist-res-state}.
This condition is also equivalent to each one of the following equalities:
\begin{align*}
&\|\chi(\Omega_{J}^\inc)\pphi_{\lambda,\e}^\circledast\|=\frac1{\sqrt2}\|\chi(\Omega^\inc)\pphi_{\lambda,\e}^\circledast\|+\ord(\e),
&&\|\chi(\Omega_{J^c}^\inc)\pphi_{\lambda,\e}^\circledast\|=\frac1{\sqrt2}\|\chi(\Omega^\inc)\pphi_{\lambda,\e}^\circledast\|+\ord(\e),
\\&
\|\chi(\Omega_J^\out)\pphi_{\lambda,\e}\|=\frac1{\sqrt2}\|\chi(\Omega^\out)\pphi_{\lambda,\e}\|+\ord(\e),
&&
\|\chi(\Omega_{J_c}^\out)\pphi_{\lambda,\e}\|=\frac1{\sqrt2}\|\chi(\Omega^\out)\pphi_{\lambda,\e}\|+\ord(\e).
\end{align*}
\end{remark}

\subsection{Comfortability}\label{sec:comf}
In recent studies of scattering problem for quantum walks, the quantity called comfortability is also investigated \cite{AKMS,HSS}.
The comfortability $\mc{E}$ of a quantum walk given by $U$ is defined for each incoming data $\alpha^\inc\in\C^{\Omega^\inc}$ and $z\in\mathbb{S}^1$ by 
\ben
\mc{E}=\mc{E}(U,\alpha^\inc,z):=\|\chi(A_0)\pphi(\alpha^\inc,z)\|_{\C^{A_0}}^2,
\een
where $\pphi$ is the unique generalized eigenfunction orthogonal to the eigenfunctions of $U$ in addition to the condition \eqref{eq:gen-ef}: $(U-z)\pphi=0$, $\chi(\Omega^\inc)\pphi=\alpha^\inc$.

Under the same setting as the previous subsection, one immediately sees
\be\label{eq:non-pert-comf}
\mc{E}(U(0),\dl_{\omega_n^\inc},z)=k_n-1,
\quad
\mc{E}(U(0),\alpha^\inc,z)= \sum_{n=1}^N|\alpha^\inc(\omega_n^\inc)|^2(k_n-1)\le\|\alpha^\inc\|^2\sup_n (k_n-1),
\ee
for $\alpha^\inc\in\C^{\Omega^\inc}$.

The following theorem shows that the comfortability grows near each resonance approaching the unit circle as $\e\to+0$.
\begin{theorem}\label{thm:comfortability}
For any fixed $z\in\mathbb{S}^1\setminus\ope{Res}(U(0))$, any $\alpha^\inc\in\C^{\Omega^\inc}$, one has
\be\label{eq:non-res-comf}
\left|\mc{E}(U(\e),\alpha^\inc,z)-\mc{E}(U(0),\alpha^\inc,z)\right|\le C\|\alpha^\inc\|^2\sqrt{\e},
\ee
as $\e\to+0$ for some $C>0$.
Fix $\lambda\in\ope{Res}(U(0))\cap\mathbb{S}^1$ such that $|\lambda_\e(\lambda)|<1$ for positive small $\e>0$.
Then one has
\ben
\mc{E}(U(\e),\til{\alpha}^\inc,\lambda_\e/|\lambda_\e|)
=\left|\lambda_\e^2\frac{1+|\lambda_\e|}{1-|\lambda_\e|}\right|
\|\chi(A_0)\pphi_{\lambda,\e}\|^2\|\chi(A_0)\pphi_{\lambda,\e}^\circledast\|^2+\ord(1)\ge
\frac{2-C\e}{1-|\lambda_\e|}+\ord(1),
\een
as $\e\to+0$, for some $C>0$ and the normalized incoming data
\ben
\til{\alpha}^\inc=\frac{1}{\|\chi(\Omega^\inc)\pphi_{\lambda,\e}^\circledast\|}\chi(\Omega^\inc)\pphi_{\lambda,\e}^\circledast.
\een
In particular, 
$\mc{E}(U(\e),\til{\alpha}^\inc,\lambda_\e/|\lambda_\e|)\ge
C\e^{-1}$
holds with some $C>0$ since $1-|\lambda_\e|=\ord(\e)$.
\end{theorem}

\section{Double and triple barrier problem on the line}\label{sec:line}
In this section, we study simple examples for quantum walks on the line.
Although the definition looks different from our setting of quantum walks on directed graphs, it can in fact isomorphic to specific cases through the isomorphism given by \eqref{eq:isom-QW}.
Computations using transfer matrices are found in Appendix~\ref{app:on-the-line}.
In a standard definition, discrete time quantum walk on the line is performed on the Hilbert space 
\ben
\ell^2(\Z;\C^2)=\{\psi:\Z\to\C^2;\,\|\psi\|_{\ell^2}<+\infty\},\quad
\|\psi\|_{\ell^2}^2=\sum_{x\in\Z}\|\psi(x)\|_{\C^2}^2.
\een
The unitary time evolution $\til{U}$ on $\ell^2(\Z,\C^2)$ is defined by the composition $\til{U}=SC$. The shift operator $S$ translates each entry of a state $\psi\in\ell^2(\Z,\C^2)$
\ben
S\psi(x)=\begin{pmatrix}1&0\\0&0\end{pmatrix}\psi(x+1)+\begin{pmatrix}0&0\\0&1\end{pmatrix}\psi(x-1),
\een
and coin operator $C$ is the multiplication operator of a $2\times2$ unitary matrix-valued function $C:\Z\to\mathrm{U}(2)$.

\subsection{Double barrier problem}\label{sec:DB}
In the double barrier problem, we assume that $C(x)=I_2$ holds except for two points $x\in \Z\setminus\{0,x_0\}$ $(x_0\ge1)$.
This is equivalent to consider a quantum walk $U$ on the directed graph $(V,A)$ given by
\begin{align*}
&V=\{L(x),\,R(x);\,x\in\Z\},\quad L(0)=R(0),\  L(x_0)=R(x_0),
\\&A=\{a_L(x):=(L(x+1),L(x)),\,a_R(x):=(R(x-1),R(x));\,x\in\Z\},
\end{align*}
and
\begin{align*}
&U\pphi(a_L(x-1))=\pphi(a_L(x)),\quad
U\pphi(a_R(x+1))=\pphi(a_R(x))\quad&&(x\notin\{0,x_0\}),\\
&\begin{pmatrix}U\pphi(a_L(x-1))\\U\pphi(a_R(x+1))\end{pmatrix}=C(x)\begin{pmatrix}\pphi(a_L(x))\\\pphi(a_R(x))\end{pmatrix}\quad&&(x\in\{0,x_0\}).
\end{align*}
One has $\iota U\iota^*=\til{U}$ with the isomorphism $\iota:\ell^2(A)\to\ell^2(\Z;\C^2)$ defined by
\be\label{eq:isom-QW}
\iota\pphi(x)=\begin{pmatrix}\pphi(a_L(x))\\\pphi(a_R(x))\end{pmatrix}.
\ee

Note that this graph has two incoming and two outgoing tails:
\begin{align*}
&V_1^\inc=\{R(x);\,x\le 0\},\quad V_2^\inc=\{L(x);\, x\ge x_0\},\\
&V_1^\out=\{L(x);\,x\le -1\},\quad V_2^\out=\{R(x);\, x\ge x_0+1\},
\end{align*}
\begin{align*}
&A_1^\inc=\{a_R(x);\,x\le 0\},\quad A_2^\inc=\{a_L(x);\,x\ge x_0\},\quad
\\&A_1^\out=\{a_L(x);\,x\le -1\},\quad A_2^\out=\{a_R(x);\,x\ge x_0+1\}.
\end{align*}
We study the transmission probability $T(z)=T(\{2\},\dl_{\omega_2^\inc},z)$ $(\omega_2^\inc=a_L(x_0))$.
It is well-known that the unitarity of the scattering matrix implies 
\ben
T(\{1\},\dl_{\omega_1^\inc},z)=T(\{2\},\dl_{\omega_2^\inc},z).
\een
An equivalent definition of the transmission probability for this case is given by 
\ben
T(z)=\left|(1\ 0)\pphi_z(N)\right|^{-2},\qquad(\text{this is invariant for }N\gg1)
\een
where $\pphi_z$ is the generalized eigenfunction $(\til{U}-z)\pphi_z=0$, such that
\be\label{eq:gen-eigen-tildeU}
\pphi_z(x)=\begin{pmatrix}z^{x}\\0\end{pmatrix}\quad(x\le-1).
\ee

In this setting, the transmission probability $T(z)$ is given by
\be\label{eq:transm-DB}
T(z)=\left|\frac{C(0)_{11}C(x_0)_{11}}{z^{2x_0}-C(0)_{21}C(x_0)_{12}}\right|^2
=1-\left|\frac{C(x_0)_{21}z^{2x_0}+C(0)_{21}(\det C(x_0))}{z^{2x_0}-C(0)_{21}C(x_0)_{12}}\right|^2,
\ee
where $C(x)_{jk}$ stands for the $(j,k)$-entry of $C(x)$.
We easily see that the existence of $z\in\mathbb{S}^1$ such that $T(z)=1$ is equivalent to 
\be\label{eq:sym-eq}
|C(0)_{21}|=|C(x_0)_{21}|.
\ee
Under this condition, $T(z)=1$ holds for $z\in\mathbb{S}^1$ satisfying
\be\label{eq:res-tun-energies}
z^{2x_0}=-\frac{C(0)_{21}(\det C(x_0))}{C(x_0)_{21}}.
\ee
On the other hand, the resonances of this quantum walk are $0$ and the $2x_0$-roots of the equation
\be\label{eq:QC-DB}
\lambda^{2x_0}=C(0)_{21}C(x_0)_{12}.
\ee
Note that each root of the equation \eqref{eq:res-tun-energies} is the closest point on $\mathbb{S}^1$ to a resonance satisfying \eqref{eq:QC-DB} since one has (see Appendix~\ref{app:on-the-line})
\be\label{eq:arg-identity}
\ope{arg}\left(-\frac{C(0)_{21}(\det C(x_0))}{C(x_0)_{21}}\right)=\ope{arg}(C(0)_{21}C(x_0)_{12}).
\ee

Let us break the condition \eqref{eq:sym-eq} slightly.
We assume that $C(x)$ depends on a small parameter $\e\ge0$, and that 
\be\label{eq:broken-sym-eq}
\bigl||C(0)_{21}|-|C(x_0)_{21}|\bigr|=\ord(\e),
\ee
as $\e\to+0$.
This implies that $|C(0)_{21}|$ and $|C(x_0)_{21}|$ are close to each other for small $\e>0$.
It however does not implies the (asymptotically perfect) resonant tunneling, that is, the existence of $z\in\mathbb{S}^1$ such that $T(z)$ is close to one.
For example, let $c_1$ and $c_2$ be positive constants and put
\be\label{eq:broken-sym-eq-ex}
C(0)=\begin{pmatrix}e^{-c_1/\e}&\sqrt{1-e^{-2c_1/\e}}\\-\sqrt{1-e^{-2c_1/\e}}&e^{-c_1/\e}\end{pmatrix},\quad
C(x_0)=\begin{pmatrix}e^{-c_2/\e}&\sqrt{1-e^{-2c_2/\e}}\\-\sqrt{1-e^{-2c_2/\e}}&e^{-c_2/\e}\end{pmatrix}.
\ee
The transmission probability is given by
\ben
T(z)=\frac{e^{-2(c_1+c_2)/\e}}{\left|z^{2x_0}-\sqrt{(1-e^{-2c_1/\e})(1-e^{-2c_2/\e})}\right|^2}\le
\frac{e^{-2(c_1+c_2)/\e}(2+\ord(e^{-6c_1/\e}+e^{-6c_2/\e}))}{(e^{-2c_1/\e}+e^{-2c_2/\e})^2}.
\een
This implies that the transmission probability $T(z)$ is exponentially small and bounded by $K \exp(-2|c_2-c_1|\e^{-1})$ with some constant $K>0$ uniformly for $z\in\mathbb{S}^1$ provided that $c_1\neq c_2$.
Note that, in this case, each resonance $\lambda$ satisfying \eqref{eq:QC-DB} is exponentially close to the unit circle for small $\e>0$:
\ben
|\lambda|=1+\ord(e^{-2c_1/\e}+e^{-2c_2/\e}).
\een
As a conclusion, neither the condition \eqref{eq:broken-sym-eq} nor the small distance from a resonance implies the resonant tunneling effect.

Recall that, in Corollary~\ref{cor:res-tun}, we do not require a resonant state to be perfectly symmetric for observing the resonant tunneling effect.
The symmetry \eqref{eq:cond-res-tun1} modulo $\ord(\e)$ results in the resonant tunneling effect modulo $\ord(\e)$: $T(z)=1+\ord(\e)$.
In the double barrier setting, this symmetry is rewritten as
\be\label{eq:DB-sym-cond}
\left|(0\ \lambda)\pphi_\lambda(x_0+1)\right|=1+\ord(\e),
\ee
where $\lambda$ is a resonance satisfying \eqref{eq:QC-DB}, and $\pphi_\lambda$ is the associated resonant state satisfying \eqref{eq:gen-eigen-tildeU}.
Since one has
\ben
\left|(0\ \lambda)\pphi_\lambda(x_0+1)\right|=
\left|\frac{C(x_0)_{21}\lambda^{x_0}+C(0)_{21}(\det C(x_0))\lambda^{-x_0}}{C(0)_{11}C(x_0)_{11}}\right|
=
\left|\frac{C(0)_{21}}{C(x_0)_{12}}\right|^{1/2}\left|\frac{C(x_0)_{22}}{C(0)_{11}}\right|,
\een
we immediately see that the example given by \eqref{eq:broken-sym-eq-ex} does not satisfy the symmetry condition \eqref{eq:DB-sym-cond} of the resonant state.

\begin{remark}
The above example given by \eqref{eq:broken-sym-eq-ex} is related to the double barrier problem of the Schr\"odinger equation on the line:
\ben
\left(-\e^2\frac{d^2}{dx^2}+V(x)-E\right)u(x)=0\quad(x\in\R),
\een
where $E>0$ and $V\in C^\infty(\R;\R)$ such that
\ben
\{x\in\R;\,V(x)\ge E\}=[a_1^-(E),a_1^+(E)]\cup[a_2^-(E),a_2^+(E)]\qquad(a_1^-<a_1^+<a_2^-<a_2^+).
\een
The correspondence between the stationary problems for Schr\"odinger equations and quantum walks on the line is studied in \cite{Hi}.
For this case, the constant $c_j$ $(j=1,2)$ is given by the Agmon distance $\int_{a_j^-}^{a_j^+}\sqrt{V(x)-E}\,dx$.
\end{remark}

\begin{remark}
There is also a positive result under \eqref{eq:broken-sym-eq}.
In addition to \eqref{eq:broken-sym-eq}, we assume that both $1-|C(0)_{21}|$ and $|C(0)_{21}|$ are bounded from below by an $h$-independent positive constant for $\e\ge0$.
Then one has $T(z)=1+\ord(\e)$ for $z\in\mathbb{S}^1$ satisfying  
\ben
z^{2x_0}=\frac{C(0)_{21}C(x_0)_{12}}{\left|C(0)_{21}C(x_0)_{12}\right|}.
\een
\end{remark}

\subsection{Triple barrier problem}\label{sec:TB}
For the triple barrier problem, we assume that $C(x)=I_2$ holds except for three points $x\in \Z\setminus\{0,x_0,x_1\}$ $(0<x_0<x_1)$.
This is equivalent to consider a quantum walk $U$ on the directed graph $(V,A)$ given by 
\begin{align*}
&V=\{L(x),\,R(x);\,x\in\Z\},\quad L(0)=R(0),\  L(x_0)=R(x_0),\ L(x_1)=R(x_1),
\\&A=\{a_L(x):=(L(x+1),L(x)),\,a_R(x):=(R(x-1),R(x));\,x\in\Z\},
\end{align*}
and
\begin{align*}
&U\pphi(a_L(x-1))=\pphi(a_L(x)),\quad
U\pphi(a_R(x+1))=\pphi(a_R(x))\quad&&(x\notin\{0,x_0,x_1\}),\\
&\begin{pmatrix}U\pphi(a_L(x-1))\\U\pphi(a_R(x+1))\end{pmatrix}=C(x)\begin{pmatrix}\pphi(a_L(x))\\\pphi(a_R(x))\end{pmatrix}\quad&&(x\in\{0,x_0,x_1\}).
\end{align*}
One has $\iota U\iota^*=\til{U}$ with the same isomorphism $\iota$ as for the double barrier case.

The transmission and reflection probabilities $T(z)$ and $R(z)=1-T(z)$ for this case is given by
\be\label{eq:transm-TB}
T(z)=\left|\frac{C(0)_{11}C(x_0)_{11}C(x_1)_{11}}{\mathtt{a}(z)}\right|^2,\qquad R(z)=\left|\frac{\mathtt{b}(z)}{\mathtt{a}(z)}\right|^2,
\ee
with $\mathtt{a}=\mathtt{a}(z)$, $\mathtt{b}=\mathtt{b}(z)$ given by
\begin{align*}
\mathtt{a}(z)
&=z^{x_1+1}-C(x_0)_{21}C(x_1)_{12}z^{2x_0-x_1+1}
\\&\qquad -C(0)_{21}C(x_0)_{12}z^{x_1-2x_0+1}-C(0)_{21}C(x_1)_{12}\det(C(x_0))z^{-x_1+1}
\\
\mathtt{b}(z)
&=
C(x_1)_{21}z^{x_1}+C(x_0)_{21}(\det C(x_1))z^{2x_0-x_1}-C(0)_{21}C(x_0)_{12}C(x_1)_{21}z^{x_1-2x_0}
\\&\qquad+C(0)_{21}(\det C(x_0)C(x_1))z^{-x_1}.
\end{align*}

We give a non-symmetric example of resonant tunneling.
Let $x_0=2$, $x_1=3$, and $z=i$.
Then the reflection coefficient $R(z)$ is
\ben
\left|\frac{-C(x_1)_{21}+C(x_0)_{21}(\det C(x_1))+C(0)_{21}C(x_0)_{12}C(x_1)_{21}+C(0)_{21}(\det C(x_0)C(x_1))}
{-1-C(x_0)_{21}C(x_1)_{12}+C(0)_{21}C(x_0)_{12}-C(0)_{21}C(x_1)_{12}(\det C(x_0))}\right|^2.
\een
Despite the position of the barriers is non-symmetric, this can vanish.
More specifically, let us assume that $C(x)$ for $x\in\{0,x_0,x_1\}$ has the following form
\be\label{eq:simple-coin-TB}
C(x)=\begin{pmatrix}\sqrt{1-r(x)^2}&r(x)\\-r(x)&\sqrt{1-r(x)^2}\end{pmatrix},\qquad(-1<r(x)<1).
\ee
Then one has
\ben
R(i)=\left|\frac{r(x_1)-r(x_0)+r(0)r(x_0)r(x_1)-r(0)}{-1+r(x_0)r(x_1)-r(0)r(x_0)+r(0)r(x_1)}\right|^2.
\een
The condition for $R(i)=0$ (equivalently $T(i)=1$) is given by
\be\label{eq:ex-TB-res-tun}
r(x_0)=\frac{r(x_1)-r(0)}{1-r(0)r(x_1)}.
\ee
The absolute value of the right-hand-side of the above identity is less than one since 
\ben
(1-r(0)r(x_1))^2-(r(x_1)-r(0))^2=(1-r(0)^2)(1-r(x_1)^2)>0.
\een
For example, $r(0)=1/2$, $r(x_0)=2/5$, $r(x_1)=3/4$ satisfy this identity.

\section{Concrete examples on graphs}\label{sec:examples}

\subsection{A model in the context of  matrix Schr\"odinger operators}\label{sec:ex-Sch}
\begin{figure}
\centering
\begin{minipage}[b]{0.49\columnwidth}
    \centering
    \includegraphics[bb=0 0 560 298, width=8cm]{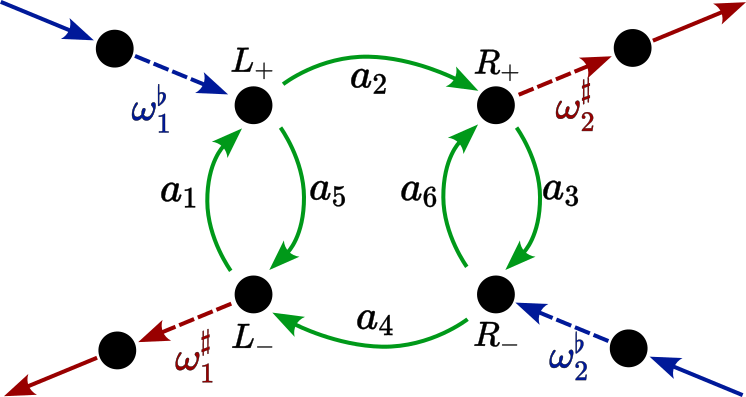}
\caption{A model related to a matrix Schr\"odinger operator}
\label{fig:matrix-Sch}
\end{minipage}
\begin{minipage}[b]{0.49\columnwidth}
    \centering
    \includegraphics[bb=0 0 412 234, width=7cm]{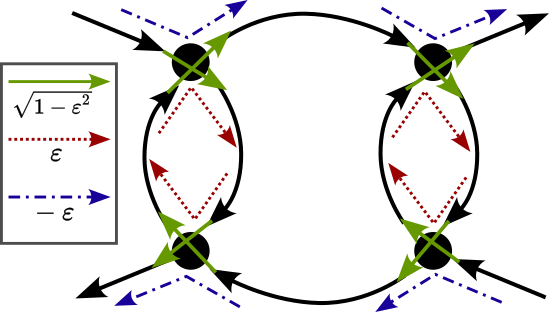}
\caption{Dynamics of the example}
\label{fig:motion_matrix-Sch}
\end{minipage}
\end{figure}
We consider a quantum walk on a graph which is related to the characteristic set of a model of matrix Schr\"odinger operator \cite{AFH1}.
The graph $(V,A)$ is introduced as follows (see also Figure~\ref{fig:matrix-Sch}).
Put $V_0:=\{L_+,R_+,L_-,R_-\}$ and $A_0:=\{a_1,\,a_2,\,\ldots,a_6\}$ with 
$a_1=(L_-,L_+)$, $a_2=(L_+,R_+)$, $a_3=(R_+,R_-)$, $a_4=(R_-,L_-)$, $a_5=(L_+,L_-)$, $a_6=(R_-,R_+)$.
Let $\omega_1^\inc$ and $\omega_2^\inc$ be arcs incoming to $L_+$ and $R_-$, respectively, and let $\omega_1^\out$ and $\omega_2^\out$ be arcs outgoing from $L_-$ and $R_+$, respectively.
Then, $V=V_0\cup V^\inc\cup V^\out$ and $A=A_0\sqcup A^\inc\sqcup A^\out$ with $V^\dir=
V_1^\dir\cup V_2^\dir$  and $A^\dir=
A_1^\dir\cup A_2^\dir$ $(\dir=\{\out,\inc\})$.

Define the isometries $U_v(\e):A_v^\inc\to A_v^\out$ for each vertex $v\in V_0$ (see Figure~\ref{fig:motion_matrix-Sch}) by 
\begin{align*}
&
U_{L_+}(\e)\begin{pmatrix}\dl_{a_1}&\dl_{\omega^\inc_1}\end{pmatrix}
=\begin{pmatrix}\dl_{a_2}&\dl_{a_5}\end{pmatrix}
\begin{pmatrix}\sqrt{1-\e^2}&-\e\\\e&\sqrt{1-\e^2}\end{pmatrix},
\\&
U_{R_-}(\e)\begin{pmatrix}\dl_{a_3}&\dl_{\omega^\inc_2}\end{pmatrix}
=\begin{pmatrix}\dl_{a_4}&\dl_{a_6}\end{pmatrix}
\begin{pmatrix}\sqrt{1-\e^2}&-\e\\\e&\sqrt{1-\e^2}\end{pmatrix},
\\&
U_{L_-}(\e)\begin{pmatrix}\dl_{a_4}&\dl_{a_5}\end{pmatrix}
=\begin{pmatrix}\dl_{a_4}&\dl_{\omega_1^\out}\end{pmatrix}
\begin{pmatrix}\sqrt{1-\e^2}&\e\\-\e&\sqrt{1-\e^2}\end{pmatrix},
\\&
U_{R_+}(\e)\begin{pmatrix}\dl_{a_2}&\dl_{a_6}\end{pmatrix}
=\begin{pmatrix}\dl_{a_3}&\dl_{\omega_2^\out}\end{pmatrix}
\begin{pmatrix}\sqrt{1-\e^2}&\e\\-\e&\sqrt{1-\e^2}\end{pmatrix}.
\end{align*}
Then for $\e=0$, one has the trivial dynamics 
\ben
U(0)\dl_{\omega_1^\inc}=\dl_{a_5},\quad
U(0)^2\dl_{\omega_1^\inc}=\dl_{\omega_1^\out},\quad
U(0)\dl_{\omega_2^\inc}=\dl_{a_6},\quad
U(0)^2\dl_{\omega_2^\inc}=\dl_{\omega_2^\out},
\een
and consequently the scattering matrix becomes diagonal
\ben
\Sigma(0,z)\begin{pmatrix}\dl_{\omega_1^\inc}&\dl_{\omega_2^\inc}\end{pmatrix}
=\begin{pmatrix}\dl_{\omega_1^\out}&\dl_{\omega_2^\out}\end{pmatrix}\begin{pmatrix}z^{-1}&0\\0&z^{-1}\end{pmatrix}.
\een
For $0\le\e<1/\sqrt{2}$, one has
\be\label{eq:s-matrix-ex1}
\Sigma(\e,z)\begin{pmatrix}\dl_{\omega_1^\inc}&\dl_{\omega_2^\inc}\end{pmatrix}
=\begin{pmatrix}\dl_{\omega_1^\out}&\dl_{\omega_2^\out}\end{pmatrix}
\frac1{z(z^2+1-2\e^2)}\begin{pmatrix}
(1-\e^2)(z^2+1)&\e^2(z^2-1)\\
\e^2(z^2-1)&(1-\e^2)(z^2+1)\end{pmatrix}.
\ee
In particular, the resonant-tunneling occurs for $z=\pm i$ (see Figure~\ref{Fig-res-tun-ms}):
\be\label{eq:res-tun-ms}
\Sigma(\e,\pm i)\begin{pmatrix}\dl_{\omega_1^\inc}&\dl_{\omega_2^\inc}\end{pmatrix}
=\begin{pmatrix}\dl_{\omega_1^\out}&\dl_{\omega_2^\out}\end{pmatrix}
\begin{pmatrix}
0&\mp i\\
\mp i&0\end{pmatrix}.
\ee
As we have seen in Corollary~\ref{cor:res-tun}, this is due to the fact that $\pm i$ are the closest points on the unit circle to the resonances $\pm i\sqrt{1-2\e^2}$.
\begin{figure}
\centering
\includegraphics[bb=0 0 1183 424, width=15cm]{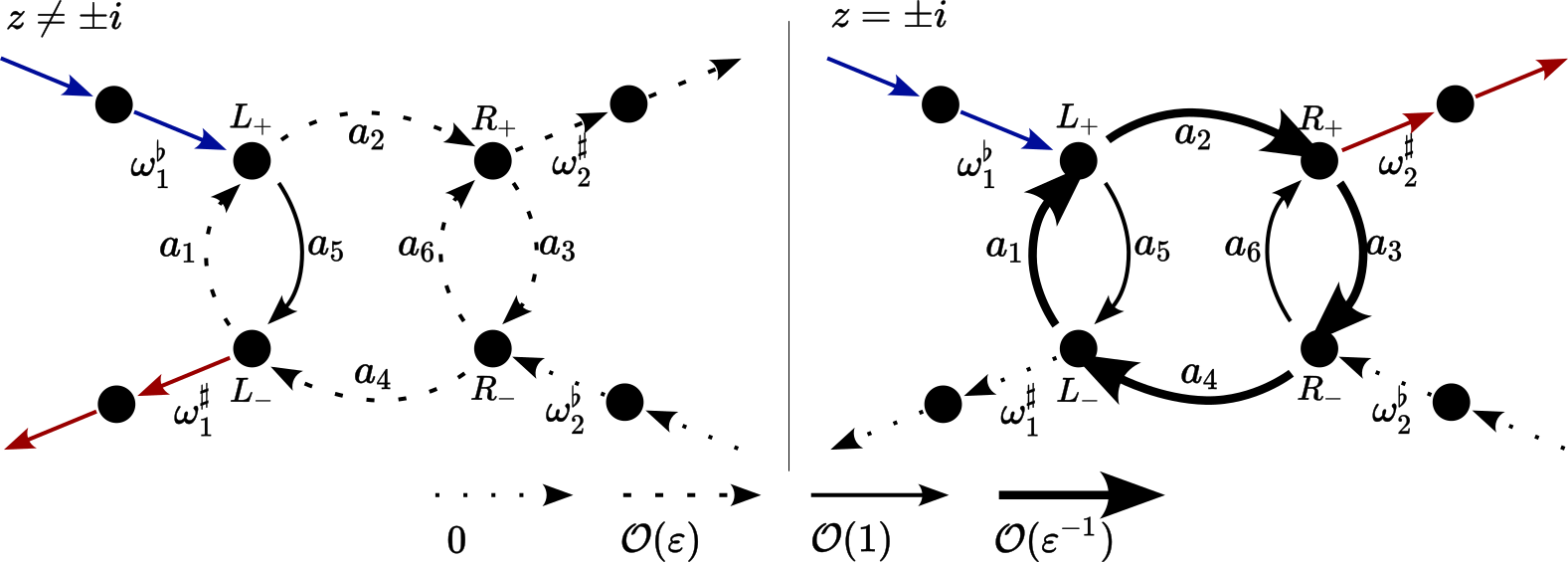}
\caption{Behavior of the generalized eigenfunction $\pphi_{\lambda_\pm}$}
\label{Fig-res-tun-ms}
\end{figure}

Let us see more informations for this model.
We use the notation
\ben
\pphi|_{A_0}=\begin{pmatrix}\pphi(a_1)&\pphi(a_2)&\pphi(a_3)&\pphi(a_4)&\pphi(a_5)&\pphi(a_6)\end{pmatrix},
\een
for a function $\pphi$ defined on $A_0$.
The eigenvalues of the matrix $U(\e)|_{A_0}=\chi(A_0)U(\e)\chi(A_0)^*$ are $\{0,\pm 1,\pm i\sqrt{1-2\e^2}\}$ $(0\le \e<1/\sqrt{2})$.
Each non-zero eigenvalue is simple while the multiplicity of the eigenvalue zero is two.
The eigenfunction $\pphi_{\pm1}$ associated with $\pm 1$ is given by
\ben
\pphi_{\pm1}|_{A_0}
=\frac12\begin{pmatrix}1&\pm\sqrt{1-\e^2}&1&\pm\sqrt{1-\e^2}&\pm\e&\pm\e\end{pmatrix},
\een
$\pphi_{\lambda_\pm}$ associated with $\lambda_\pm=\pm i\sqrt{1-2\e^2}$ is given by
\ben
\pphi_{\lambda_\pm}|_{A_0}
=\frac1{2\lambda_\pm}\begin{pmatrix}\lambda_\pm&\sqrt{1-\e^2}&-\lambda_\pm&-\sqrt{1-\e^2}&\e&-\e\end{pmatrix},
\een
and $\pphi_{0,1}$ and $\pphi_{0,2}$ associated with  $0$ is given by
\ben
\pphi_{0,1}|_{A_0}=\begin{pmatrix}0&0&0&-\e&\sqrt{1-\e^2}&0\end{pmatrix},\qquad
\pphi_{0,2}|_{A_0}=\begin{pmatrix}0&-\e&0&0&0&\sqrt{1-\e^2}\end{pmatrix}.
\een
Then the corresponding eigenfunctions to $(U(\e)|_{A_0})^*$ are given by $\pphi_{\pm1}^\circledast=\pphi_{\pm1}$, 
\ben
\pphi_{\lambda_\pm}^\circledast|_{A_0}
=\frac1{2\lambda_\pm}\begin{pmatrix}\lambda_\pm&\sqrt{1-\e^2}&-\lambda_\pm&-\sqrt{1-\e^2}&-\e&\e\end{pmatrix},
\een
\ben
\pphi_{0,1}^\circledast|_{A_0}=
\frac1{1-2\e^2}\begin{pmatrix}0&-\e(1-\e^2)&0&\e^3&(1-\e^2)^{3/2}&-\e^2\sqrt{1-\e^2}\end{pmatrix},
\een
and
\ben
\pphi_{0,2}^\circledast|_{A_0}=\frac1{1-2\e^2}
\begin{pmatrix}0&\e^3&0&-\e(1-\e^2)&-\e^2\sqrt{1-\e^2}&(1-\e^2)^{3/2}\end{pmatrix}.
\een
Note that $\chi(A_0)^*\pphi_{\pm1}$ is an eigenfunction associated with the eigenvalue $\pm1$ of $U(\e)$. 
We use the same notation $\pphi_{\lambda_\pm}$ for the resonant state obtained by continuing the eigenfunction $\pphi_{\lambda_\pm}$ (Proposition~\ref{prop:matrix-reduction}).
Similarly, let $\pphi_{\lambda_\pm}^\circledast$ denote the incoming resonant state obtained by continuing $\pphi_{\lambda_\pm}^\circledast$ (Proposition~\ref{prop:matrix-reduction2}).
Then 
\begin{align*}
&\pphi_{\lambda_\pm}(\omega_1^\out)=\frac{\e\sqrt{1-\e^2}}{1-2\e^2},\qquad
\pphi_{\lambda_\pm}(\omega_2^\out)=-\frac{\e\sqrt{1-\e^2}}{1-2\e^2},
\\&
\pphi_{\lambda_\pm}^\circledast(\omega_1^\inc)=-\frac{\e\sqrt{1-\e^2}}{1-2\e^2},\qquad
\pphi_{\lambda_\pm}^\circledast(\omega_2^\inc)=\frac{\e\sqrt{1-\e^2}}{1-2\e^2}.
\end{align*}
As we will see in Thoerem~\ref{thm:res-exp}, the scattering matrix is given by
\begin{align*}
\Sigma(\e,z)\alpha^\inc
&=\sum_{\pm}\frac{\lambda_\pm^2(\alpha^\inc,\chi(\Omega^\inc)\pphi_{\lambda_\pm}^\circledast)}{z-\lambda_\pm}\chi(\Omega^\out)\pphi_{\lambda_\pm}+\sum_{k=1,2}\frac{(\alpha^\inc,\chi(\Omega^\inc)U(\e)^*\pphi_{0,k}^\circledast)}{z}\chi(\Omega^\out)U(\e)\pphi_{0,k}
\\&
=\frac{-2\e\sqrt{1-\e^2}(\alpha^\inc(\omega_1^\inc)-\alpha^\inc(\omega_2^\inc))}{z^2+1-2\e^2}\left(z\re\chi(\Omega^\out)\pphi_{\lambda_+}-\sqrt{1-2\e^2}\im\chi(\Omega^\out)\pphi_{\lambda_+}\right)
\\&\quad+\frac{(1-\e^2)\alpha^\inc(\omega_1^\inc)-\e^2\alpha^\inc(\omega_2^\inc)}{(1-2\e^2)z}\dl_{\omega_1^\out}
+\frac{-\e^2\alpha^\inc(\omega_1^\inc)+(1-\e^2)\alpha^\inc(\omega_2^\inc)}{(1-2\e^2)z}\dl_{\omega_2^\out}
\\&
=\frac1{z(z^2+1-2\e^2)}\Bigl((z^2+1)(1-\e^2)
(\alpha^\inc(\omega_1^\inc)\dl_{\omega_1^\out}
+\alpha^\inc(\omega_2^\inc)\dl_{\omega_2^\out})
\\&\qquad\qquad\qquad\qquad\qquad+\e^2(z^2-1)
(\alpha^\inc(\omega_1^\inc)\dl_{\omega_2^\out}
+\alpha^\inc(\omega_2^\inc)\dl_{\omega_1^\out})
\Bigr).
\end{align*}
This coincides with \eqref{eq:s-matrix-ex1}.
Moreover, the generalized eigenfunction is given by $\pphi_{\alpha^\inc,z}=\alpha^\inc(\omega_1^\inc)\pphi_{1,z}+\alpha^\inc(\omega_2^\inc)\pphi_{2,z}$ with
\begin{align*}
\pphi_{1,z}
=&\frac1{(z^2+1-2\e^2)}
\Bigl(z\e\sqrt{1-\e^2}\quad
-\e(z^2-\e^2)\quad
-z\e\sqrt{1-\e^2}
\\&\qquad\qquad\qquad
-\e(1-\e^2)\quad
\sqrt{1-\e^2}(z^2+1-\e^2)\quad 
-\e^2\sqrt{1-\e^2}\Bigr),
\end{align*}
and
\begin{align*}
\pphi_{2,z}
=&\frac1{(z^2+1-2\e^2)}
\Bigl(-z\e\sqrt{1-\e^2} \quad 
-\e(1-\e^2)\quad 
z\e\sqrt{1-\e^2}
\\&\qquad\qquad\qquad
-\e(z^2-\e^2)\quad 
-\e^2\sqrt{1-\e^2}\quad 
\sqrt{1-\e^2}(z^2+1-\e^2)\Bigr),
\end{align*}
in $A_0$, and 
\begin{align*}
&\pphi_{1,z}(\omega_1^\inc)=1,\quad
\pphi_{1,z}(\omega_2^\inc)=0, \qquad
\pphi_{2,z}(\omega_1^\inc)=0,\quad
\pphi_{2,z}(\omega_2^\inc)=1,
\\&\pphi_{1,z}(\omega_1^\out)=\frac{(1-\e^2)(z^2+1)}{z^2(z^2+1-2\e^2)},\quad
\pphi_{1,z}(\omega_2^\out)=\frac{\e^2(z^2-1)}{z^2(z^2+1-2\e^2)}, 
\\&
\pphi_{2,z}(\omega_1^\out)=\frac{\e^2(z^2-1)}{z^2(z^2+1-2\e^2)},\quad
\pphi_{2,z}(\omega_2^\out)=\frac{(1-\e^2)(z^2+1)}{z^2(z^2+1-2\e^2)}.
\end{align*} 

\subsection{A model of cycle graph}\label{sec:model-cycle}
\begin{figure}
\centering
\begin{minipage}[b]{0.49\columnwidth}
    \centering
    \includegraphics[bb=0 0 550 311, width=8cm]{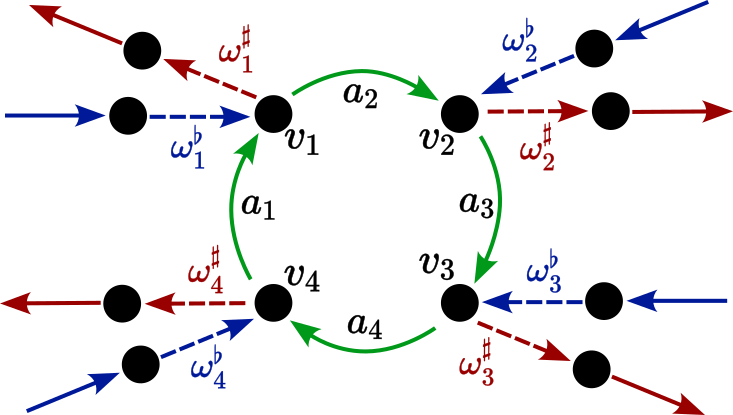}
\caption{An example of cycle graph with $N=4$}
\label{fig:cycle}
\end{minipage}
\begin{minipage}[b]{0.49\columnwidth}
    \centering
    \includegraphics[bb=0 0 312 150, width=6cm]{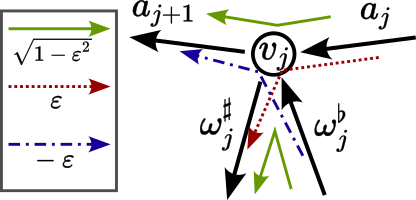}
\caption{Dynamics of the example}
\label{fig:motion_cycle}
\end{minipage}
\end{figure}
We then consider another model (see Figure~\ref{fig:cycle}).
Fix an integer $N\ge2$.
Let $V_0:=\{v_1,v_2,\ldots,v_N\}$ and $A_0:=\{a_1,a_2,\ldots,a_N\}$ with 
$a_n=(v_{n-1},v_n)$ $(v_0=v_N)$.
Put $\Omega^\inc:=\{\omega_1^\inc,\ldots,\omega_N^\inc\}$ and $\Omega^\out:=\{\omega_1^\out,\ldots,\omega_N^\out\}$, where $\omega_n^\inc$ is an arc incoming to $v_n$, and $\omega_n^\out$ is an arc outgoing from $v_n$ $(n=1,2,\ldots,N)$.

Define the isometry $U_v(\e):A_v^\inc\to A_v^\out$ for each vertex $v\in V_0$ by
\ben
U_{v_n}(\e)\begin{pmatrix}\dl_{a_n}&\dl_{\omega_n^\inc}\end{pmatrix}
=\begin{pmatrix}\dl_{a_{n+1}}&\dl_{\omega_n^\out}\end{pmatrix}
\begin{pmatrix}\sqrt{1-c_n^2\e^2}&c_n\e\\ -c_n\e&\sqrt{1-c_n^2\e^2}\end{pmatrix},
\een
with a constant $c_n>0$ (see Figure~\ref{fig:motion_cycle}).
Then for $\e=0$, one has the trivial dynamics
\ben
U(0)\dl_{\omega_n^\inc}=\dl_{\omega_n^\out}\quad(n=1,2,\ldots,N),
\een
and consequently the scattering matrix becomes diagonal
\ben
\Sigma(0,z)\begin{pmatrix}\dl_{\omega_1^\inc}&\ldots&\dl_{\omega_N^\inc}\end{pmatrix}
=\begin{pmatrix}\dl_{\omega_1^\out}&\ldots&\dl_{\omega_N^\out}\end{pmatrix}I_N.
\een
For positive small $\e$, one has
\be\label{eq:s-matrix-ex2}
\Sigma(\e,z)\dl_{\omega_n^\inc}(\omega_l^\out)
=\left\{
\begin{aligned}
&\frac{(1-c_n^2\e^2)z^N-\tau_N}{\sqrt{1-c_n^2\e^2}(z^N-\tau_N)}
&&(l=n)\\
&-\frac{c_nc_l\e^2\tau_{l-1}\tau_N}{\tau_n(z^N-\tau_N)}z^{n-l}
&&(l<n)\\
&-\frac{c_nc_l\e^2\tau_{l-1}}{\tau_n(z^N-\tau_N)}z^{N+n-l}
&&(l>n)
\end{aligned}\right.
\ee
with $\tau_n=\prod_{k=1}^n\sqrt{1-c_k^2\e^2}$.
In particular, the resonant-tunneling effect takes place, for example, in the following situations (see Figure~\ref{Fig-res-tun-cycle}).
\begin{figure}
\centering
\includegraphics[bb=0 0 1798 752, width=15cm]{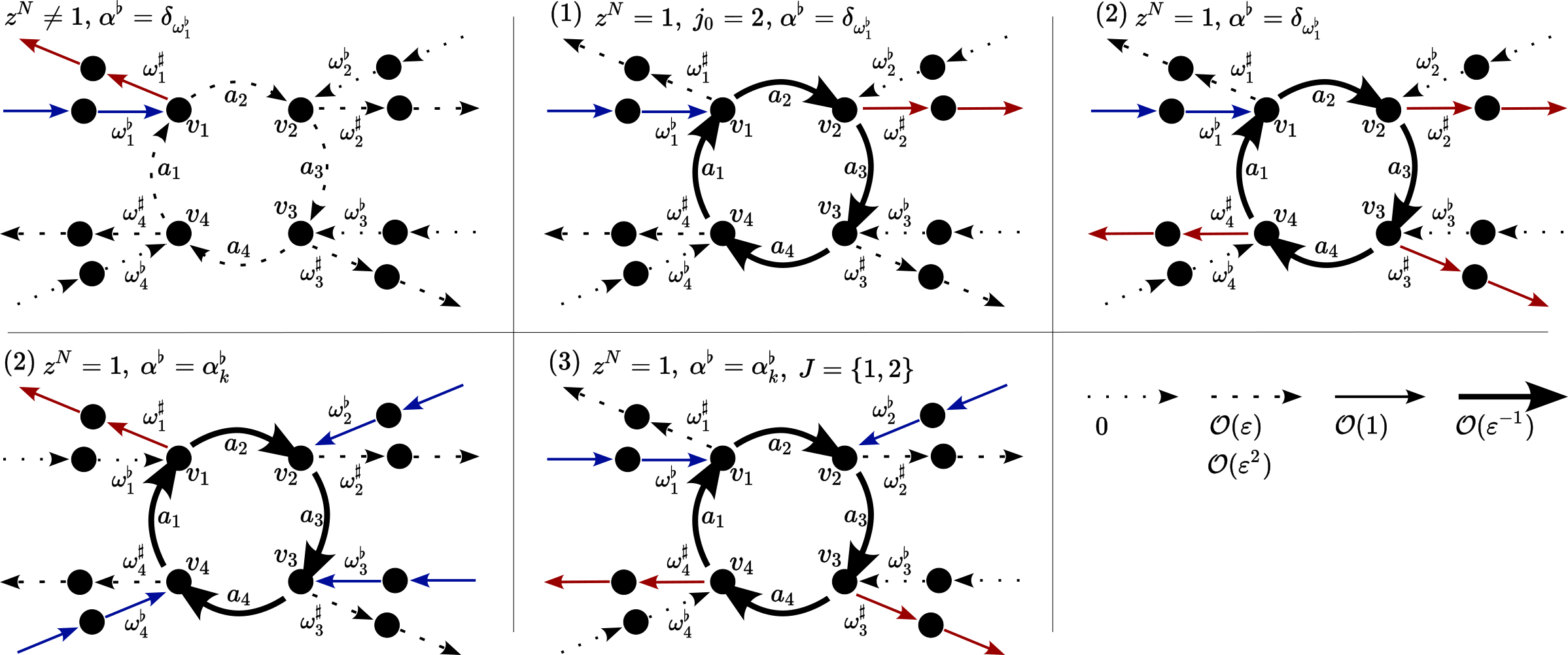}
\caption{Behavior of the generalized eigenfunction in each situation}
\label{Fig-res-tun-cycle}
\end{figure}
\begin{enumerate}
\item There exist $j_0\in\{2,3,\ldots,N\}$ such that $c_{j_0}=c_1\neq0$ and $c_n=0$ for $n\notin\{1,j_0\}$. Then one has $\tau_N=1-c_1\e^2$, 
\begin{align*}
&\Sigma(\e,z)\dl_{\omega_{1}^\inc}=\frac{(z^N-1)\sqrt{1-c_1^2\e^2}}{z^N-1+c_1^2\e^2}\dl_{\omega_1^\out}
-\frac{c_1^2\e^2z^{N+1-j_0}}{(z^N-1+c_1^2\e^2)}\dl_{\omega_{j_0}^\out},
\\&\Sigma(\e,z)\dl_{\omega_{j_0}^\inc}=\frac{(z^N-1)\sqrt{1-c_1^2\e^2}}{z^N-1+c_1^2\e^2}\dl_{\omega_{j_0}^\out}
-\frac{c_1^2\e^2z^{j_0-1}}{(z^N-1+c_1^2\e^2)}\dl_{\omega_{1}^\out},
\end{align*}
and $\Sigma(\e,z)\dl_{\omega_n^\inc}=\dl_{\omega_n^\out}$ for $n\notin\{1,j_0\}$.
In particular, for $z\in\mathbb{S}^1$ with $z^N=1$, one has
\be\label{eq:res-tun-c1}
\Sigma(\e,z)\dl_{\omega_1^\inc}=-z^{1-j_0}\dl_{\omega_{j_0}^\out},\quad
\Sigma(\e,z)\dl_{\omega_{j_0}^\inc}=-z^{j_0-1}\dl_{\omega_1^\out}.
\ee
\item 
Suppose $\sum_{n=2}^Nc_n^2=c_1^2\neq0$.
Put 
\ben
\alpha_k^\inc=\sum_{n=2}^N\frac{c_n\tau^{\frac{n-1}N}}{c_1\tau_n}e^{-i2\pi k(n-1)/N}\dl_{\omega_n^\inc},
\een
for $k=0,1,\ldots,N-1$.
Then one has $\tau_N=1-c_1^2\e^2+\ord(\e^4)$, and 
\be\label{eq:res-tun-c2}
\Sigma(\e,z)\dl_{\omega_{1}^\inc}(\omega_l^\out)=\left\{
\begin{aligned}
&\ord(\e^2)&&(l=1)\\
&-\frac{c_l}{c_1z^{l-1}}+\ord(\e^2)&&(l\ge2),
\end{aligned}\right.
\ee
\be\label{eq:res-tun-c3}
\Sigma(\e,e^{i2\pi k/N})\alpha_k^\inc=-\dl_{\omega_1^\out}+\ord(\e^2),
\ee
for $z\in\mathbb{S}^1$ with $z^N=1$.
\item 
There exists $J\subset\{1,2,\ldots,N\}$ such that $\sum_{n\in J}c_n^2=\sum_{n\notin J}c_n^2=:c^2\neq0$.
Then for 
\ben
\alpha^\inc_k:=\sum_{n\in J}\frac{c_n\tau_N^{\frac{n-1}N}}{c\tau_n}e^{-i2\pi k(n-1)/N}\dl_{\omega_n^\inc},
\een
one has
\be\label{eq:res-tun-c4}
\Sigma(\e,e^{2\pi ki/N})\alpha^\inc_k=-\sum_{n\notin J}c_n\dl_{\omega_n^\out}+\ord(\e^2).
\ee
\end{enumerate}

Let us see more informations for this model.
The non-zero resonances are $\{z\in\C;\, z^N=\tau_N\}=\{\lambda_k:=\tau_N^{1/N}\mu^k;\,k=1,2,\ldots,N\}$ with $\mu=e^{2\pi  i/N}$. They are all simple.
An associated resonant state $\pphi_{\lambda_k}$ and corresponding incoming resonant state $\pphi_{\lambda_k}^\circledast$ are given by
\ben
\pphi_{\lambda_k}(a_l)
=\frac{\tau_{l-1}}{\sqrt{N} \lambda_k^{l-1}},\quad
\pphi_{\lambda_k}^\circledast(a_l)
=\frac{\overline{\lambda_k}^{l-1}}{\sqrt{N}\tau_{l-1}},\quad(l=1,2,\ldots, N),
\een
on $A_0$, and
\ben
\pphi_{\lambda_k}(\omega_l^\out)=-\frac{\e c_l\tau_{l-1}}{\sqrt{N}\lambda_k^l},\quad
\pphi_{\lambda_k}^\circledast(\omega_l^\inc)=\frac{\e c_l\overline{\lambda_k}^{l-1}}{\sqrt{N}\tau_l},\quad(l=1,2,\ldots, N).
\een
As we will see in Thoerem~\ref{thm:res-exp}, the scattering matrix is given by
\begin{align*}
\Sigma(\e,z)\dl_{\omega_n^\inc}(\omega_l^\out)
&=\sum_{k=1}^N\frac{\lambda_k^2\overline{\pphi_{\lambda_k}^\circledast(\omega_n^\inc)}}{z-\lambda_k}\pphi_{\lambda_k}(\omega_l^\out)
+\til{\dl}_{n,l}
\\&
=-\frac{\e^2c_nc_l\tau_{l-1}\tau_N^{(n-l-1)/N}}{N\tau_n}
\sum_{k=1}^N\frac{\mu^{k(n-l-1)}}{z-\tau_N^{1/N}\mu^k}
+\til{\dl}_{n,l}.
\end{align*}
Formula~\eqref{eq:s-matrix-ex2} follows from this with
\be\label{eq:elementary-pw}
\sum_{k=0}^{N-1}\frac{\mu^{kp}}{z-c\mu^k}
=\frac{Nc^{N-p}z^{p-1}}{z^N-c^N}.
\qquad(c>0,\ p=1,2,\ldots,N)
\ee
See Appendix~\ref{app:elementary} for a proof of this identity.

\section{Scattering matrix in terms of resonant states}\label{sec:S-matrix-res}
In this section, we give an expression (Theorem~\ref{thm:res-exp}) of the scattering matrix $\Sigma(z)$ of a quantum walk induced by a unitary operator $U$.
We neither consider a family of operators depending on some parameter nor assume the simplicity of eigenvalues (i.e., we discuss without Condition~\ref{cond:pert}).
For this purpose, we introduce some notions of generalized resonant states.

\subsection{Statement of the result}
An outgoing function $\varphi\in\C^A$ is said to be a \textit{generalized resonant state} associated with a resonance $\lambda\in\ope{Res}(U)\setminus\{0\}$ if there exists $k\ge1$ such that $(U-\lambda)^k\varphi=0$.
Each generalized resonant state $\pphi$ satisfies $(U-\lambda)^{m(\lambda)}\pphi=0$ since $m(\lambda)$ is by definition the dimension of the vector space of the associated generalized resonant states $W(\lambda)$. 
As in the general theory, we can take a basis of $W(\lambda)$ which consists of Jordan chains.
For each non-zero resonance $\lambda\in\ope{Res}(U)\setminus\{0\}$, there exist $K=K(\lambda)\in\{1,2,\ldots,m(\lambda)\}$, numbers $l_1=l_1(\lambda),\ldots,l_K=l_K(\lambda)\in\{1,2,\ldots,m(\lambda)\}$ with $\sum_{k=1}^Kl_k= m(\lambda)$ and linearly independent outgoing functions $\{\pphi_{\lambda,k,l};\,k\in\{1,\ldots,K\},\,l\in\{1,\ldots,l_k\}\}$ such that
\ben
\pphi_{\lambda,k,l}=(U-\lambda)^{l_k-l}\pphi_{\lambda,k,l_k}\neq0\qtext{for}l=0,1,\ldots,l_k-1,\quad
(U-\lambda)^{l_k}\pphi_{\lambda,k,l_k}=0,
\een
holds for $k=1,2,\ldots,K$.
For each $k$, the tuple $(\pphi_{\lambda,k,l})_{1\le l\le l_k}$ is called \textit{a Jordan chain}.

The \textit{(generalized) incoming resonant states}  and \textit{Jordan chains} are also defined accordingly.
The following facts are consequence of  Propositions~\ref{prop:matrix-reduction}, \ref{prop:matrix-reduction2} and general theory of matrices (see e.g., \cite{Kato}).
For each basis $\{\pphi_{\lambda,k,l};\,k\in\{1,\ldots,K\},\,l\in\{1,\ldots,l_k\}\}$ of $W(\lambda)$ consists of $K$-Jordan chains, there exists a basis $\{\pphi^\circledast_{\lambda,k,l};\,k\in\{1,\ldots,K\},\,l\in\{1,\ldots,l_k\}\}$ of the space of generalized incoming resonant states associated with $\bar{\lambda}^{-1}$ such that
\be\label{eq:in-res1}
\pphi_{\lambda,k,l}^\circledast
=(U^*-\bar{\lambda})^{l-1}
\pphi_{\lambda,k,1}^\circledast\neq0\quad\text{for }\ l=1,2,\ldots,l_k,\quad
(U^*-\bar{\lambda})^{l_k}\pphi_{\lambda,k,1}^\circledast=0,
\ee
for each $k$, and
\be\label{eq:in-res2}
(\pphi_{\lambda,k,l},\pphi^\circledast_{\lambda,k',l'})
=\sum_{a\in A}\pphi_{\lambda,k,l}(a)\overline{\pphi^\circledast_{\lambda,k',l'}(a)}
=\til{\dl}_{(k,l),(k',l')},
\ee
where $\til{\dl}$ stands for the Kronecker delta.
Define an operator-valued function $M_{\lambda}=M_{\lambda}(z):\C^{\Omega^\inc}\to\C^{\Omega^\out}$ by 
\be\label{eq:def-M_lam2}
M_{\lambda}(z)\alpha^\inc=
\sum_{k=1}^{K(\lambda)}\sum_{l=1}^{l_k(\lambda)}
\left(\sum_{p=0}^{l_k-l}\frac{\left(\alpha^\inc,\chi(\Omega^\inc)
U^{-2}\pphi^\circledast_{\lambda,k,l+p}\right)_{\C^{\Omega^\inc}}}{(z-\lambda)^{p+1}}\right)
\chi(\Omega^\out)\pphi_{\lambda,k,l}\qquad(\alpha^\inc\in\C^{\Omega^\inc}).
\ee
\begin{remark}
Note that one has
\ben
U^{-2}\pphi_{\lambda,k,l+p}^\circledast
=\bar{\lambda}^2\pphi_{\lambda,k,l+p}^\circledast
+2\bar\lambda\pphi_{\lambda,k,l+p+1}^\circledast
+\pphi_{\lambda,k,l+p+2}^\circledast
\een
for each $k,l,p$ with the convention $\pphi_{\lambda,k,l}=0$ for $l>l_k$.
This implies that
\ben
\left(\alpha^\inc,\chi(\Omega^\inc)
U^{-2}\pphi^\circledast_{\lambda,k,l+p}\right)_{\C^{\Omega^\inc}}
=\sum_{\omega\in\Omega^\inc}\left(\lambda^2\overline{\pphi_{\lambda,k,l+p}(\omega)}
+2\lambda\overline{\pphi_{\lambda,k,l+p+1}(\omega)}+\overline{\pphi_{\lambda,k,l+p+2}(\omega)}\right)
\alpha^\inc(\omega).
\een
\end{remark}

Concerning the zero-resonance, we introduce the matrix $U|_{A_0}:=\chi(A_0)U\chi(A_0)^*$ acting on $\C^{A_0}$.
Suppose that $z=0$ is an eigenvalue of this matrix.
Then there exist $K=K(0)\in\N$, numbers $l_1=l_1(0),\ldots,l_K=l_K(0)$, and a basis  $\{v_{0,k,l};\,k\in\{1,\ldots,K\},\,l\in\{1,\ldots,l_k\}\}$ of the generalized eigenspace such that each $(v_{0,k,l})_{1\le l\le l_k}$ forms a Jordan chain.
There also exist a basis $\{v^\circledast_{0,k,l};\,k\in\{1,\ldots,K\},\,l\in\{1,\ldots,l_k\}\}$ of the generalized eigenspace associated with the eigenvalue $z=0$ of $(U|_{A_0})^*$, such that one has
\begin{align*}
&\chi(A_0)\pphi_{0,k,l}=\chi(A_0)U^{l_k-l}\pphi_{0,k,l_k}\neq0,\quad
\chi(A_0)U^{l_k}\pphi_{0,k,l_k}=0,
\\&
\chi(A_0)\pphi_{0,k,l}^\circledast=\chi(A_0)U^{-(l-1)}\pphi_{0,k,1}^\circledast\neq0,\quad
\chi(A_0)U^{-l_k}\pphi_{0,k,1}^\circledast=0,
\end{align*}
and $(\pphi_{0,k,l},\pphi^\circledast_{0,k',l'})=\til{\dl}_{(k,l),(k',l')}$ for $\pphi_{0,k,l},$ $\pphi_{0,k,l}^\circledast\in\cH$ given by
\ben
\pphi_{0,k,l}:=\chi(A_0)^*v_{0,k,l},\qquad
\pphi_{0,k,l}^\circledast:=\chi(A_0)^*v_{0,k,l}^\circledast.
\een
Define an operator-valued function $M_0=M_0(z):\C^{\Omega^\inc}\to\C^{\Omega^\out}$ by
\be\label{eq:def-M_0}
M_0(z)\alpha^\inc=
\sum_{k=1}^{K(0)}\sum_{l=1}^{l_k(0)}\left(\sum_{p=0}^{l_k(0)-l}\frac1{z^{p+1}}\left(\alpha^\inc,\chi(\Omega^\inc)U^*\pphi_{0,k,l+p}^\circledast\right)\right)\chi(\Omega^\out)U\pphi_{0,k,l}
+\chi(\Omega^\out)U\chi(\Omega^\inc)^*\alpha^\inc.
\ee
Remark that the operator $\chi(\Omega^\out)U\chi(\Omega^\inc)^*$ does not vanish identically only if there exists a pair $(\omega^\inc,\omega^\out)\in\Omega^\inc\times\Omega^\out$ such that $ (\omega^\inc)^+=(\omega^\out)^-=:v$ and $(U_v)_{\omega^\out,\omega^\inc}\neq0$.

The definition of the generalized eigenfunctions and the scattering matrix is extended to $z\in(\C\setminus\ope{Res}(U))\cup\mathbb{S}^1$ (Proposition~\ref{prop:gen-ef}). 
Here, the scattering matrix $\Sigma(z)$ is no longer isometric for $z\notin\mathbb{S}^1$.

The following theorem gives the resonance expansion of the scattering matrix.
\begin{theorem}\label{thm:res-exp}
For each $z\in(\C\setminus\ope{Res}(U))\cup\mathbb{S}^1$, one has
\ben
\Sigma(z)=\sum_{\lambda\in\ope{Res}(U)\setminus\mathbb{S}^1}M_\lambda(z).
\een
\end{theorem}

\subsection{Reformation as matrix problems}\label{sec:reduce-to-matrix}
Let us put 
\ben
U|_{A_0}:=\chi(A_0)U\chi(A_0)^*:\C^{A_0}\to\C^{A_0}.
\een
We show three propositions which reduces the problems of $U$ on $\cH$ to that of $U|_{A_0}$ on $\C^{A_0}$.

\begin{lemma}\label{lem:on-tails}
Let $\varphi\in\C^A$ be a solution to the equation $(U-z)\varphi=0$ for some $z\in\C\setminus\{0\}$.
Then there exist $\alpha^\inc\in\C^{\Omega^\inc}$ and $\alpha^\out\in\C^{\Omega^\out}$ such that
\be\label{eq:on-tails}
\pphi(a)=\left\{
\begin{aligned}
&\alpha^\inc(\omega_n^\inc)z^j&&a=a_{n,j}^\inc\ (j\ge0)\\
&\alpha^\out(\omega_n^\out)z^{-j}&&a=a_{n,j}^\out\ (j\ge1).
\end{aligned}\right.
\ee
Consequently, $\pphi$ is an outgoing resonant state (resp. incoming resonant state, eigenfunction) if and only if $\alpha^\inc=0$ (resp. $\alpha^\out=0$, $\alpha^\inc=\alpha^\out=0$).
Similarly, each generalized resonant state vanishes on $A^\inc$, and each generalized incoming resonant state does on $A^\out$.
\end{lemma}

\begin{remark}
Note that $U$ is unitary, in particular normal: $U^*U=UU^*$.
This results the semi-simplicity of the eigenvalues, that is, the generalized eigenspace coincides with the eigenspace.
\end{remark}

\begin{proof}
According to \eqref{eq:dyn-on-tails}, at each arc of $A^\inc$ and of $A^\out$, the identity $(U-z)\varphi=0$ is rewritten as
\ben
\pphi(a_{n,j+1}^\inc)=z\pphi(a_{n,j}^\inc),\quad
\pphi(a_{n,j-1}^\out)=z\pphi(a_{n,j}^\out).
\een
This inductively gives \eqref{eq:on-tails} with the initialization $\alpha^\inc(\omega_n^\inc):=\pphi(a_{n,0}^\inc)$ and $\alpha^\out(\omega_n^\out):=z\pphi(a_{n,1}^\out)$.
The condition that $\varphi$ to be outgoing or incoming is now trivial.
Let a function $\pphi$ be an eigenfunction. 
Then it belongs to $\cH$, that is, the norm 
\ben
\|\varphi\|_{\cH}^2=\sum_{a\in A}|\pphi(a)|^2\ge\sum_{a\in A\setminus A_0}|\pphi(a)|^2
=\sum_{n=1}^N\left(|\alpha^\inc(\omega_n^\inc)|^2\sum_{j\ge0}|z|^{2j}
+
|\alpha^\out(\omega_n^\out)|^2\sum_{j\ge1}|z|^{-2j}\right)
\een
is finite.
Note that the operator $U$ is unitary, hence every eigenvalue lies on the unit circle $\mathbb{S}^1$.
Therefore, $\pphi\in\cH$ only if $\alpha^\inc=\alpha^\out=0$.

Let $\pphi$ be a generalized resonant state such that $(U-\lambda)^k\pphi=0$ with $\lambda\neq0$.
Then for each arc $a_{n,j}^\inc\in A^\inc$, one has
\ben
(U-\lambda)^k\pphi(a_{n,j}^\inc)
=\sum_{l=0}^k\begin{pmatrix}k\\l\end{pmatrix}(-\lambda)^{k-l}\pphi(a_{n,j+l}^\inc)
=(-\lambda)^{k}\pphi(a_{n,j}^\inc)+\sum_{l=1}^k\begin{pmatrix}k\\l\end{pmatrix}(-\lambda)^{k-l}\pphi(a_{n,j+l}^\inc).
\een 
This shows that if $\pphi(a_{n,j}^\inc)\neq0$, then at least for one of $l=1,2,\ldots,k$, one has $\pphi(a_{n,j+l}^\inc)\neq0$.
This with the fact that $\pphi$ is outgoing, one concludes that $\pphi(a_{n,j}^\inc)=0$ for any $j\ge0$.

Similarly let $\varphi$ be a generalized incoming resonant state such that $(U-\lambda)^k\pphi=0$ with $\lambda\neq0$.
This is equivalent to $(U^*-\lambda^{-1})^k\pphi=(-\lambda^{-1}U^*)^k(U-\lambda)^k\pphi=0$.
This shows 
\ben
(U^*-\lambda^{-1})^k\pphi(a_{n,j}^\out)
=\sum_{l=0}^k\begin{pmatrix}k\\l\end{pmatrix}(-\lambda^{-1})^k\pphi(a_{n,j+l}^\out),
\een
for each $a_{n,j}^\out\in A^\out$, 
and consequently $\pphi(a_{n,j}^\out)=0$ for any $j\ge1$.
\end{proof}

For each subset $B$ and each $k\ge1$, put 
\ben
\mc{N}_k(B):=\{a\in A;\,\text{there exists a  path of length $k+1$ from an arc of $B$ to $a$}\},
\een
\ben
\mc{N}_{-k}(B):=\{a\in A;\,\text{there exists a  path of length $k+1$ from $a$ to an arc of $B$}\},
\een
and $\mc{N}_0(B):=B$.
Here, a finite sequence $(a_1,a_2,\ldots,a_k)$ of arcs $a_j\in A$ is said to be a \textit{path of length} $k$ \textit{from $a_1$ to $a_k$} if $a_j^+=a_{j+1}^-$ holds for each $j=1,2,\ldots,k-1$.
In general, one has $\mc{N}_{-k}(\mc{N}_k(B))\supset B$.

\begin{lemma}\label{lem:nbd}
Let $B\subset A$ and $k\in\Z$.  One has
\be\label{eq:nbd}
\ope{supp}\,(U^k\mathbbm{1}_B\pphi)\subset \mc{N}_k(B),
\ee
for any function $\pphi\in\C^A$.
One also has
\be\label{eq:intertwine}
U^k\mathbbm{1}_B\pphi=\mathbbm{1}_{\mc{N}_k(B)}U^k\pphi,
\ee
provided that $\mc{N}_{-k}(\mc{N}_k(B))= B$.
Especially, we will frequently use
\be\label{eq:inter}
U\mathbbm{1}_{A_0\cup\Omega^\inc}\pphi
=\mathbbm{1}_{A_0\cup\Omega^\out}U\pphi.
\ee
\end{lemma}

\begin{proof}
For any $b\in B$, one has
\ben
U\dl_{b}=\sum_{a\in A_{b^+}^\out}(U_{b^+})_{a,b}\dl_a.
\een
Note here that $\mc{N}_1(\{b\})=A_{b^+}^\out$. In fact, $(b,a)$ is a path of length two if and only if $b^+=a^-$, equivalently $a\in A_{b^+}^\out$.
Then, for any function $\pphi\in\C^A$, one has
\ben
U\mathbbm{1}_B\pphi
=
U\sum_{b\in B}\pphi(b)\dl_b
=\sum_{b\in B}\pphi(b)\sum_{a\in A^\out_{b^+}}(U_{b^+})_{a,b}\dl_a.
\een
This means $\ope{supp} (U\mathbbm{1}_B\pphi)\subset \bigcup_{b\in B}A_{b^+}^\out=\mc{N}_1(B)$.

Let $k\ge1$, suppose that $\supp(U^k\mathbbm{1}_B\pphi)\subset\mc{N}_k(B)$.
Then one has
\ben
U^{k+1}\mathbbm{1}_B\pphi
=U\sum_{b\in\mc{N}_k(B)}(U^k\mathbbm{1}_B\pphi)(b)\dl_b
=\sum_{b\in\mc{N}_k(B)}(U^k\mathbbm{1}_B\pphi)(b)\sum_{a\in A_{b^+}^\out}(U_{b^+})_{a,b}\dl_a.
\een
Then one obtains
\ben
\ope{supp}\,(U^{k+1}\mathbbm{1}_B\pphi)
\subset\bigcup_{b\in\mc{N}_k(B)} A_{b^+}^\out
=\mc{N}_{k+1}(B).
\een
Formula \eqref{eq:nbd} is now proven for each positive $k$ by induction.
We can also prove \eqref{eq:nbd} for $k\le-1$ in the parallel argument starting with the identity
\ben
U^{-1}\dl_b=\sum_{a\in A^\inc_{b^-}}(U_{b_-}^*)_{a,b}\dl_a.
\een

As a consequence of \eqref{eq:nbd}, one has
\ben
U^{k}\mathbbm{1}_B\pphi
=\mathbbm{1}_{\mc{N}_k(B)}U^{k}\mathbbm{1}_B\pphi
=\mathbbm{1}_{\mc{N}_k(B)}U^{k}\pphi
-\mathbbm{1}_{\mc{N}_k(B)}U^{k}\mathbbm{1}_{A\setminus B}\pphi,
\een
and
$U^k\mathbbm{1}_{A\setminus B}\pphi
=\mathbbm{1}_{\mc{N}_k(A\setminus B)}U^k\mathbbm{1}_{A\setminus B}\pphi$.
It suffices to show that the sets $\mc{N}_k(A\setminus B)$ and $\mc{N}_k(B)$ are disjoint under the condition $\mc{N}_{-k}(\mc{N}_k(B))=B$.
By definition, the initial arc of any path of length $k+1$ to an arc of $\mc{N}_k(B)$ belongs to $\mc{N}_{-k}(\mc{N}_k(B))=B$. 
This implies that the terminus 
of any path of length $k+1$ starting from an arc of $A\setminus B$ is outside of $\mc{N}_k(B)$, i.e., $\mc{N}_k(A\setminus B)\cap\mc{N}_k(B)=\emptyset$.
\end{proof}

The proposition below gives us the correspondence between outgoing resonances and eigenvalues of $U|_{A_0}$.

\begin{proposition}\label{prop:matrix-reduction}
\begin{enumerate}
\item\label{enu:red1} The non-zero resonances of $U$ and non-zero eigenvalues of $U|_{A_0}$ coincide including their multiplicities. 
\item\label{enu:red2} For each generalized resonant state $\pphi$ associated with $\lambda\in\ope{Res}(U)\setminus\{0\}$, $\chi(A_0)\varphi\in\C^{A_0}$ is a generalized eigenvector of $U|_{A_0}$.
\item\label{enu:red3} Every resonance lies inside or on the unit circle $\mathbb{S}^1$.
Resonances lying on the unit circle are eigenvalues of $U$.
\item\label{enu:red4} For any resonant state $\pphi$ associated with $\lambda\in\ope{Res}(U)\setminus\{0\}$, one has
\be\label{eq:widths}
\frac{\|\mathbbm{1}_{\Omega^\out}\pphi\|^2}{\|\mathbbm{1}_{A_0}\pphi\|^2}=|\lambda|^{-2}-1.
\ee
\end{enumerate}
\end{proposition}

\begin{proof}
Remark that for any function $\psi\in\C^A$, one has
\ben
U|_{A_0}\chi(A_0)\psi
=
\chi(A_0)U\chi(A_0)^*\chi(A_0)\psi
=
\chi(A_0)U\mathbbm{1}_{A_0}\psi.
\een
$(\ref{enu:red1},\ref{enu:red2})$ Let $\varphi\in\C^A$ be a generalized resonant state associated with a non-zero resonance $\lambda\in\ope{Res}(U)\setminus\{0\}$ such that $(U-\lambda)^k\varphi=0$.
According to Lemma~\ref{lem:on-tails}, $\pphi$ vanishes on $A^\inc$, in particular on $\Omega^\inc$.
It follows from 
\eqref{eq:inter} that 
\ben
\chi(A_0)U\pphi
=\chi(A_0)\mathbbm{1}_{A_0\cup\Omega^\out}U\pphi
=\chi(A_0)U\mathbbm{1}_{A_0\cup\Omega^\inc}\pphi
=\chi(A_0)U\mathbbm{1}_{A_0}\pphi
=U|_{A_0}\chi(A_0)\pphi.
\een
By applying the above identity repeatedly, we obtain
\ben
(
U|_{A_0}-\lambda)^k\chi(A_0)\varphi
=\chi(A_0)(U-\lambda)^k\varphi=0,
\een
that is, $U\pphi$ is again a generalized eigenfunction.
This means that the multiplicity of $\lambda$ as an eigenvalue of $
U|_{A_0}$ is larger than or equal to $m(\lambda)$.

Conversely, let $u$ be a generalized eigenvector of 
$U|_{A_0}$ associated with $\lambda\in\C\setminus\{0\}$ satisfying $(
U|_{A_0}-\lambda)^ku=0$.
This can be extended to a generalized resonant state $\til{u}$ which is defined by
\ben
\til{u}=\chi(A_0)^*u+\sum_{n=1}^N\sum_{l=1}^k\alpha_l^\out(\omega_n^\out)\chi(A_n^\out)^*F_{l,\lambda,n},
\een
where $\alpha_l^\out(\omega_n^\out):=(U-\lambda)^l\chi(A_0)^*u(\omega_n^\out)$, and $F_{l,\lambda,n}\in\C^{A_n^\out}$ defined by $F_{l,\lambda,n}(a_{n,j}):=\lambda^{-j-l+1}f_l(j)$ with 
$f_l(j)$ defined inductively for $l\ge1$ and $j\ge1$ by 
\ben
f_1(j)=1,\qquad
f_l(j)=\sum_{k=1}^jf_{l-1}(k).
\een
Then one can check that $(U-\lambda)^k\til{u}=0$ by using the fact that $(U-\lambda)^k\chi(A_n^\out)^*F_{l,\lambda,n}=-\dl_{\omega_n^\out}$.

\noindent
$(\ref{enu:red3})$ Let $\varphi$ be a resonant state associated with $\lambda\in\ope{Res}(U)\setminus\{0\}$.
Then the unitarity of $U$ acting on $\cH$ and 
the identity \eqref{eq:inter} imply
\begin{align*}
\left\|\mathbbm{1}_{A_0}\varphi\right\|
=
\left\|U\mathbbm{1}_{A_0}\varphi\right\|
=
\left\|U\mathbbm{1}_{A_0\cup\Omega^\inc}\varphi\right\|
=\left\|\mathbbm{1}_{A_0\cup\Omega^\out}U\pphi\right\|
=\left\|\lambda\mathbbm{1}_{A_0\cup\Omega^\out}\pphi\right\|.
\end{align*}
Since $A_0$ and $\Omega^\out$ are disjoint, it follows that
\be\label{eq:ineq-out}
|\lambda|^2\left\|\mathbbm{1}_{A_0}\pphi\right\|^2
=
\left\|\mathbbm{1}_{A_0}\varphi\right\|^2
-|\lambda|^2\left\|\mathbbm{1}_{\Omega^\out}\pphi\right\|^2.
\ee
This shows $|\lambda|\le1$.
Moreover, the equality $|\lambda|=1$ holds if and only if
\be\label{eq:on-omega-out}
\mathbbm{1}_{\Omega^\out}\pphi
=\chi(\Omega^\out)^*\alpha^\out=0.
\ee
This implies that $\pphi$ is an eigenfunction.

\noindent
$(\ref{enu:red4})$ Finally one obtains \eqref{eq:widths} by combining \eqref{eq:ineq-out} and \eqref{eq:on-omega-out}. 
\end{proof}

The following proposition for incoming resonances is almost parallel with Proposition~\ref{prop:matrix-reduction} for outgoing resonances.

\begin{proposition}\label{prop:matrix-reduction2}
\begin{enumerate}
\item Every incoming resonance of $U$ lies outside or on the unit circle $\mathbb{S}^1$. Incoming resonances lying on the unit circle are eigenvalues of $U$.
\item The reciprocal of each incoming resonance is an eigenvalue of $U|_{A_0}^*$, and the multiplicity of each incoming resonance $\lambda$ coincides with that of an eigenvalue $\lambda^{-1}$ of $U|_{A_0}^*$.
\item For each generalized eigenvector $u$ associated with an eigenvalue $\lambda\in\C\setminus\{0\}$ of $U|_{A_0}^*$, there exists a generalized incoming resonant state $\pphi^\inc$ of the incoming resonance $\lambda^{-1}$ such that $\chi(A_0)\pphi^\inc=u$. 
\item For any incoming resonant state $\pphi^\inc$ associated with an incoming resonance $\lambda\in\ope{Res}^\inc(U)$, one has
\be\label{eq:esti-omega-in}
\frac{\|\mathbbm{1}_{\Omega^\inc}\pphi^\inc\|^2}{\|\mathbbm{1}_{A_0}\pphi^\inc\|^2}=|\lambda|^2-1.
\ee
\end{enumerate}
\end{proposition} 

\begin{remark}
Let $u_k$ be a generalized eigenvector associated with an eigenvalue $\lambda$ of $U|_{A_0}^*$ such that $(U|_{A_0}^*-\lambda)^ku_k=0$.
Then for $\pphi^\inc_k\in\C^A$ such that $\chi(A_0)\pphi^\inc_k=u_k$, one has
\ben
(U^*-\lambda)^k\pphi^\inc_k=0,
\een 
which is equivalent to $(U-\lambda^{-1})^k\pphi^\inc_k=0$.
The vector $u_{k-1}:=(U|_{A_0}^*-\lambda)u_k$ is again a generalized eigenvector associated with $\lambda$ provided that $u_{k-1}\neq0$.
Then for the corresponding generalized incoming resonant state $\pphi^\inc_{k-1}$, one has
\ben
(U^*-\lambda)\pphi^\inc_k=\pphi^\inc_{k-1},
\quad\text{and }\ 
(U-\lambda^{-1})\pphi^\inc_k=-\lambda^{-1}U\pphi^\inc_{k-1}.
\een
This shows that the structure of the Jordan chains for the generalized incoming resonant states of $U$ are not exactly same as that for the generalized eigenfunctions of $U|_{A_0}^*$.
One has
\ben
(U-\lambda^{-1})^l\til{\pphi}_k^\inc
=\til{\pphi}_{k-l}^\inc,\qquad
\til{\pphi}_{l}^\inc:=(-\lambda^{-1}U)^{k-l}\pphi_l^\inc.
\een
\end{remark}

The above proposition is a consequence of the following lemma with Proposition~\ref{prop:matrix-reduction}.
We prepare some terminologies of the graph theory.
The \textit{reverse graph}, also known as \textit{transposed graph}, of a directed graph $(V,A)$ is the directed graph $(V,A^T)$, where every arc is reversed.
More precisely, the arc set $A^T$ is the set of inverted arcs of the original graph $A^T=\{ a^{-1};\, a\in A\}$, where $(a^{-1})^+=a^-$, $(a^{-1})^-=a^+$.
Define $\tau:\C^{A^T}\to\C^{A}$ by
\ben
\tau\pphi(a)=\pphi(a^{-1})\qquad(\pphi\in\C^{A^T}).
\een

\begin{lemma}
The operator $\tau^{-1}U^*\tau$ is a quantum walk on the reverse graph $(V,A^T)$ induced by the isometries $U_v^*$.
Each (generalized) incoming resonant state $\pphi^\inc$ associated with $\lambda\in\ope{Res}^\inc(U)$ defines a (generalized) outgoing resonant state $\tau^{-1}\pphi^\inc$ associated with $\lambda^{-1}\in\ope{Res}(\tau^{-1}U^*\tau)\setminus\{0\}$, and vice versa.
\end{lemma}

\begin{proof}
The incoming and outgoing tails of the original graph is reversed in the reversed graph.
In particular, $\tau\pphi\in\C^A$ is incoming provided that $\pphi\in\C^{A^T}$ is outgoing.
Therefore, for each generalized resonant state $\pphi$ of $\tau^{-1}U^*\tau$ associated with $\lambda\in\ope{Res}(\tau^{-1}U^*\tau)$ such that $(\tau^{-1}U^*\tau-\lambda)^k\pphi=0$, one has 
\ben
\tau^{-1}(-\lambda U^*)^k(U-\lambda^{-1})^k\tau \pphi
=\tau^{-1}(U^*-\lambda)^k\tau \pphi
=(\tau^{-1}U^*\tau-\lambda)^k\pphi=0.
\een
The invertibility of $\tau^{-1}(-\lambda U^*)^k$ implies that $\tau\pphi$ is a generalized incoming resonant state associated with $\lambda^{-1}$.
\end{proof}

We finally see the construction of generalized eigenfunctions by using the (reduced) resolvent of $U|_{A_0}$.
Note that, if we construct them by using the resolvent of $U$, we need to show the so-called limiting absorption principle which guarantees the existence of the limit in some sense of the resolvent as the spectral parameter approaches the essential spectrum $\mathbb{S}^1$ (see e.g., \cite{Mo}).
Moreover, we need to introduce meromorphic continuation of the resolvent to enter the inside of the unit circle (see e.g., \cite{HMS}).

\begin{proposition}\label{prop:gen-ef}
For any $\alpha^\inc\in\C^{\Omega^\inc}$ and any $z\in(\C\setminus\ope{Res}(U))\cup\mathbb{S}^1$, there exists a function $\varphi\in\C^A$ satisfying \eqref{eq:gen-ef}: $(U-z)\pphi=0$ with $\alpha^\inc=\chi(\Omega^\inc)\pphi$.
Such a function is unique if $z$ is not an eigenvalue of $U$.
The function $\alpha^\out:=z\chi(\Omega^\out)\varphi\in\C^{\Omega^\out}$ is uniquely determined by $\alpha^\inc$ and $z\in(\C\setminus\ope{Res}(U))\cup\mathbb{S}^1$ even if $z$ is an eigenvalue.
Moreover, the scattering matrix $\Sigma(z)$ is isometric for $z\in\mathbb{S}^1$.
\end{proposition}

This proposition is a version of the Rellich theorem (see e.g., \cite[Theorems 3.33, 3.35]{DyZw}).
A part of the proof below is based on the method of \cite[proof of Lemma 3.5]{HiSe} where they studied the Grover walk on graphs with tails.
\begin{proof}
We first show that each eigenvalue $\lambda$ of modulus one of $U|_{A_0}$ is semi-simple. 
Let $u\in\C^{A_0}$ satisfy $(U|_{A_0}-\lambda)^ku=0$ for $k\ge2$ with $\lambda\in\mathbb{S}^1$.
Then $v:=(U|_{A_0}-\lambda)^{k-1}u$ is an eigenvector: $(U|_{A_0}-\lambda)v=0$, and consequently $\chi(A_0)^*v$ is an eigenfunction of $U$ (see the construction of $\til{u}$ in the proof of Proposition~\ref{prop:matrix-reduction}). 
Then by the unitarity of $U$, one obtains 
\be\label{eq:v_is_ef}
(U^*-\lambda^{-1})\chi(A_0)^*v=(-\lambda^{-1}U^*)(U-\lambda)\chi(A_0)^*v=0.
\ee
This implies $(U|_{A_0}^*-\lambda^{-1})v=\chi(A_0)(U^*-\lambda^{-1})\chi(A_0)^*v=0$, and
\ben
\|(U|_{A_0}-\lambda)^{k-1}u\|^2
=((U|_{A_0}-\lambda)^{k-1}u,v)
=((U|_{A_0}-\lambda)^{k-2}u,(U|_{A_0}^*-\bar{\lambda})v)=0.
\een
We obtain $(U|_{A_0}-\lambda)^{k-1}u=0$, and $(U|_{A_0}-\lambda)u=0$ by repeating this process.
This gives the semi-simplicity of each eigenvalue.

Put $f=f(\alpha^\inc):=\chi(A_0)U\chi(\Omega^\inc)^*\alpha^\inc\in\C^{A_0}$.
Then $f$ is orthogonal to any eigenvector $v$ associated with any eigenvalue $\lambda$ of modulus one of $U|_{A_0}$.
Recall that \eqref{eq:v_is_ef} shows $\chi(A_0)^*v$ is an eigenfunction of $U^*$, and
\ben
(v,f)_{\C^{A_0}}=(U^*\chi(A_0)^*v,\chi(\Omega^\inc)^*\alpha^\inc)_{\cH}
=\lambda^{-1}(\chi(A_0)^*v,\chi(\Omega^\inc)^*\alpha^\inc)_{\cH}=0.
\een
Here, we have used the fact that $\chi(\Omega^\inc)^*\alpha^\inc$ is supported only on $\Omega^\inc$.
Since each eigenvalue $\lambda$ of modulus one of $U|_{A_0}$ is semi-simple, one has $P_\lambda f=0$ for the projection onto the (generalized) eigenspace associated with $\lambda$ defined by
\be\label{eq:def-proj}
P_\lambda=-\frac1{2\pi i}\oint_{\lambda}(U|_{A_0}-z)^{-1}dz.
\ee
Let us put $m(0)=\ope{rank}P_0$, the multiplicity of the zero eigenvalue of $U|_{A_0}$.
By a general theory of matrices, one has
\ben
(U|_{A_0}-z)^{-1}g=-\sum_{\lambda\in\ope{Res}(U)}\left(\sum_{k=0}^{m(\lambda)-1}\frac{(U|_{A_0}-\lambda)^kP_\lambda g}{(z-\lambda)^{k+1}}\right),
\een
for $z\in\C\setminus\ope{Res}(U)$.
This with the orthogonality $P_\lambda f=0$ for each $\lambda\in\ope{Res}(U)\cap\mathbb{S}^1$ implies 
\ben
(U|_{A_0}-z)u(\alpha^\inc,z)=-f
\een
for $z\in(\C\setminus\ope{Res}(U))\cup\mathbb{S}^1$ with
\be\label{eq:local-resolvent}
u(\alpha^\inc,z):=
\sum_{\lambda\in\ope{Res}(U)\setminus\mathbb{S}^1}\left(\sum_{k=0}^{m(\lambda)-1}\frac{(U|_{A_0}-\lambda)^kP_\lambda f}{(z-\lambda)^{k+1}}\right).
\ee
Put $\alpha^\out:=\chi(\Omega^\out)U(\chi(A_0)^*u(z,\alpha^\inc)+\chi(\Omega^\inc)^*\alpha^\inc)\in\C^{\Omega^\out}$.
Then one can easily check that
\be\label{eq:construct-gen-ef}
\varphi_{\alpha^\inc,z}(a):=\left\{
\begin{aligned}
&u(\alpha^\inc,z;a)&&a\in A_0\\
&\alpha^\inc(\omega_n)z^j&&a=a_{n,j}^\inc\in A^\inc\\
&\alpha^\out(\omega_n)z^{-j}&&a=a_{n,j}^\out\in A^\out
\end{aligned}\right.
\ee
is a generalized eigenfunction.

Suppose that there exist two functions $\varphi_1$ and $\varphi_2$ on $A$ such that $\chi(\Omega^\inc)\varphi_1=\chi(\Omega^\inc)\varphi_2$ and $(U-z_0)\varphi_1=(U-z_0)\varphi_2=0$ for some $z_0\in\C$.
It then follows that $\varphi_1-\varphi_2$ is outgoing, and that $z_0$ is a resonance if $\varphi_1-\varphi_2$ is non-trivial.
This shows that the uniqueness of the generalized eigenfunction is broken when $z_0$ is a resonance.
However, the outgoing data $\alpha^\out$ is invariant even if $z_0$ is an eigenvalue.
In this case, $\pphi_1-\pphi_2$ is an eigenfunction, and its support is contained in $A_0$, as we have seen in Lemma~\ref{lem:on-tails}.

According to the identity \eqref{eq:inter}, one obtains for $z\in\mathbb{S}^1$,
\begin{align*}
\|\alpha^\inc\|^2+\|u(\alpha^\inc,z)\|^2
&=\|\mathbbm{1}_{A_0\cup\Omega^\inc}\pphi_{\alpha^\inc,z}\|^2
=\|U\mathbbm{1}_{A_0\cup\Omega^\inc}\pphi_{\alpha^\inc,z}\|^2
=\|\mathbbm{1}_{A_0\cup\Omega^\out}U\pphi_{\alpha^\inc,z}\|^2
\\&=\|\alpha^\out\|^2+\|u(\alpha^\inc,z)\|^2.
\end{align*}
This shows the isometry of the scattering matrix $\Sigma(z)$.
\end{proof}

\subsection{Proof of the resonance expansion}\label{sec:prf-res-exp}
We prove Theorem~\ref{thm:res-exp} in this section.
As we have seen in the proof of  Proposition~\ref{prop:gen-ef}, the generalized eigenfunctions can be given by the formulas \eqref{eq:local-resolvent} and \eqref{eq:construct-gen-ef} by using the resolvent of the matrix $U|_{A_0}=\chi(A_0)U\chi(A_0)^*$.
By the general theory of matrices, each projection $P_\lambda$ defined by \eqref{eq:def-proj} onto the generalized eigenspace of $\lambda$ is rewritten as
\be\label{eq:proj-exp}
P_\lambda f=\sum_{k=1}^K\sum_{l=1}^{l_k}
(f,v_{\lambda,k,l}^\circledast)v_{\lambda,k,l},
\ee
where $\{v_{\lambda,k,l};\,k\in\{1,\ldots,K(\lambda)\},\,l\in\{1,\ldots,l_k(\lambda)\}\}$ and $\{v_{\lambda,k,l}^\circledast;\,k\in\{1,\ldots,K(\lambda)\},\,l\in\{1,\ldots,l_k(\lambda)\}\}$ are generalized eigenfunctions associated with the eigenvalue $\lambda$ of $U|_{A_0}$ and with the eigenvalue $\bar\lambda$ of $U|_{A_0}^*$, respectively.
They are characterized by the conditions 
\ben
(v_{\lambda,k,l},v_{\lambda,k',l'}^\circledast)_{\C^{A_0}}=\til{\dl}_{(k,l),(k',l')},
\een
and
\ben
(U|_{A_0}-\lambda)^{l_k-l}v_{\lambda,k,l_k}=v_{\lambda,k,l},\qquad
(U|_{A_0}^*-\bar\lambda)^{l-1}v_{\lambda,k,1}^\circledast=v_{\lambda,k,l}^\circledast,
\een
for $l\in\{1,\ldots,l_k\}$ with $v_{k,0}=v_{k,l_k+1}^\circledast=0$.
Let $\pphi_{\lambda,k,l}$ be the generalized resonant state associated with the resonance $\lambda\in\ope{Res}(U)$ such that $\chi(A_0)\pphi_{\lambda,k,l}=v_{\lambda,k,l}$ and let $\pphi_{\lambda,k,l}^\circledast$ be the generalized incoming resonant state associated with the incoming resonance $\bar\lambda^{-1}$ such that $\chi(A_0)\pphi_{\lambda,k,l}^\circledast=v_{\lambda,k,l}^\circledast$ (see Proposition~\ref{prop:matrix-reduction2} for the existence).
Then, one has
\ben
(f,v^\circledast_{\lambda,k,l})
=(\alpha^\inc,\chi(\Omega^\inc)U^*\chi(A_0)^*v^\circledast_{\lambda,k,l})
=\left(\alpha^\inc,\chi(\Omega^\inc)U^*\pphi_{\lambda,k,l}^\circledast\right),
\een
for $f=\chi(A_0)U\chi(\Omega^\inc)^*\alpha^\inc$ with $\alpha^\inc\in\C^{\Omega^\inc}$.
By substituting this into \eqref{eq:proj-exp}, we obtain
\ben
P_\lambda f=\sum_{k=1}^K\sum_{l=1}^{l_k}\left(\alpha^\inc,\chi(\Omega^\inc)U^*\pphi^\circledast_{\lambda,k,l}\right)\chi(A_0)\pphi_{\lambda,k,l},
\een
and
\begin{align*}
(U|_{A_0}-\lambda)^pP_\lambda f
&=\sum_{k=1}^K\sum_{l=1}^{l_k}\left(\alpha^\inc,\chi(\Omega^\inc)U^*\pphi^\circledast_{\lambda,k,l}\right)\chi(A_0)\pphi_{\lambda,k,l-p}
\\&
=\sum_{k=1}^K\sum_{l=1}^{l_k-p}\left(\alpha^\inc,\chi(\Omega^\inc)U^*\pphi^\circledast_{\lambda,k,l+p}\right)\chi(A_0)\pphi_{\lambda,k,l}.
\end{align*}

When $z=0$ is an eigenvalue of $U|_{A_0}$, we put
\ben
\pphi_{0,k,l}:=\chi(A_0)^*v_{0,k,l},\qquad
\pphi^\circledast_{0,k,l}:=\chi(A_0)^*v^\circledast_{0,k,l}.
\een
Then $\pphi_{0,k,l}$ is not a generalized resonant state, and $\pphi^\circledast_{0,k,l}$ is not a generalized incoming resonant state.
However, the above computation is valid also for $\lambda=0$.
Let $u(z,\alpha^\inc)$ be the function defined by \eqref{eq:local-resolvent}:
\ben
u(z,\alpha^\inc)
=
\sum_{\lambda\in\ope{Res}(U)\setminus\mathbb{S}^1}\left(\sum_{p=0}^{m(\lambda)-1}\frac{(U|_{A_0}-\lambda)^p P_\lambda f}{(z-\lambda)^{p+1}}\right).
\een
This is now rewritten as
\begin{align*}
u(z,\alpha^\inc)
&=\sum_{\lambda\in\ope{Res}(U)\setminus\mathbb{S}^1}
\sum_{p=0}^{m(\lambda)-1}\frac1{(z-\lambda)^{p+1}}
\sum_{k=1}^K\sum_{l=1}^{l_k-p}\left(\alpha^\inc,\chi(\Omega^\inc)U^*\pphi^\circledast_{\lambda,k,l+p}\right)\chi(A_0)\pphi_{\lambda,k,l}
\\&=\sum_{\lambda\in\ope{Res}(U)\setminus\mathbb{S}^1}
\sum_{k=1}^K\sum_{l=1}^{l_k}
\sum_{p=0}^{l_k-l}\frac1{(z-\lambda)^{p+1}}\left(\alpha^\inc,\chi(\Omega^\inc)U^*\pphi^\circledast_{\lambda,k,l+p}\right)\chi(A_0)\pphi_{\lambda,k,l}
\\&=:\sum_{\lambda\in\ope{Res}(U)\setminus\mathbb{S}^1}u_\lambda(z,\alpha^\inc).
\end{align*}
This gives a resonance expansion in $A_0$ of the generalized eigenfunction $\pphi_{\alpha^\inc,z}$ defined by \eqref{eq:construct-gen-ef} and \eqref{eq:local-resolvent}.

According to the formula right before \eqref{eq:construct-gen-ef}, 
the outgoing data $\alpha^\out$ is given by $\alpha^\out=\chi(\Omega^\out)U(\chi(A_0)^*u(z,\alpha^\inc)+\chi(\Omega^\inc)^*\alpha^\inc)$.
For each $\lambda\in\ope{Res}(U)\setminus(\mathbb{S}^1\cup\{0\})$, 
we deduce
\begin{align*}
\chi(\Omega^\out)U\chi(A_0)^*u_\lambda(z,\alpha^\inc)
=M_\lambda\alpha^\inc,
\end{align*}
from
$
U\chi(A_0)^*\chi(A_0)\pphi_{\lambda,k,l}
=U\mathbbm{1}_{A_0}\pphi_{\lambda,k,l}
=\mathbbm{1}_{A_0\cup\Omega^\out}\left(\lambda\pphi_{\lambda,k,l}+\pphi_{\lambda,k,l-1}\right)
$ and $U^*\pphi_{\lambda,k,l}^\circledast=\bar\lambda\pphi_{\lambda,k,l}^\circledast+\pphi_{\lambda,k,l+1}^\circledast$ (see \eqref{eq:def-M_lam2} for the definition of $M_\lambda$).
One also has (see \eqref{eq:def-M_0} for the definition of $M_0$)
\begin{align*}
\chi(\Omega^\out)U\chi(A_0)^*u_0(z,\alpha^\inc)
&=\sum_{k=1}^{K(0)}\sum_{l=1}^{l_k(0)}
\sum_{p=0}^{l_k(0)-l}\frac1{z^{p+1}}\left(\alpha^\inc,\chi(\Omega^\inc)U^*\pphi^\circledast_{0,k,l+p}\right)\chi(\Omega^\out)U\pphi_{0,k,l}
\\&=\left(M_0-\chi(\Omega^\out)U\chi(\Omega^\inc)\right)\alpha^\inc.
\end{align*}
This ends the proof of Theorem~\ref{thm:res-exp}.

\section{Proof of the asymptotic behavior}\label{sec:prf-main}
In this section, we prove Theorem~\ref{thm:non-resonant}, Corollary~\ref{cor:res-tun} stated in Subsection~\ref{sec:results} and Theorem~\ref{thm:comfortability} done in Subsection~\ref{sec:comf}.
\subsection{Resolvent estimates}
We prepare some estimates on reduced resolvents.
Put $E(\e):=\ope{EV}(U(\e)|_{A_0})\cap\{z\in\C;\,|z|\le1/2\}$, where $\ope{EV}(U(\e)|_{A_0})$ stands for the set of eigenvalues of $U(\e)|_{A_0}$. 
Let $P_\e$ be the projection onto the union $\bigoplus_{z\in E(\e)}\ope{Ker}\left((U(\e)|_{A_0}-z)^{\#A_0}\right)$ of generalized eigenspaces associated with eigenvalues belonging to $E(\e)$ of $U(\e)|_{A_0}$.

We study the reduced resolvents
\ben
\til{R}(\e,z):=R(\e,z)P_{\e}=-\sum_{\lambda\in E(\e)}\left(\sum_{k=0}^{m(U(\e)|_{A_0};\lambda)-1}\frac{(U(\e)|_{A_0}-z)^kP_{\lambda,\e}}{(z-\lambda)^{k+1}}\right),
\een
\ben
Q(\e,z):=-\sum_{\lambda\in \ope{Res}(U(0))\cap\mathbb{S}^1,\ |\lambda_\e(\lambda)|\neq1}\frac{P_{\lambda_\e(\lambda),\e}}{z-\lambda_\e(\lambda)},
\een
where $P_{\lambda,\e}$ stands for the projection onto the generalized eigenspace associated with the eigenvalue $\lambda$ of $U(\e)|_{A_0}$.
Recall that $\lambda_\e(\lambda)$ denotes the unique eigenvalue of $U(\e)|_{A_0}$ approaching  $\lambda$ in the limit $\e\to+0$.

Note that one has $Q(0,z)=0$ since $\ope{EV}(U(0)|_{A_0})\subset\ope{Res}(U(0))\subset\mathbb{S}^1\cup\{0\}$.
Roughly speaking, the difference $Q(\e,z)-Q(0,z)=Q(\e,z)\not\to0$ as $\e\to+0$ makes some discrepancy, while $\til{R}(\e,z)-\til{R}(0,z)\to0$ as $\e\to+0$.
\begin{proposition}\label{prop:res-esti}
As $\e\to+0$, one has
\ben
\|\til{R}(\e,z)-\til{R}(0,z)\|_{\C^{A_0}\to\C^{A_0}}=\ord(\e).
\een
\end{proposition}

To prove this, we prepare the following lemma.
The assumption on the simplicity of eigenvalues of $U(0)$ is used here.
\begin{lemma}\label{lem:proj-estimate}
There exists $C>0$ such that
\ben
\sup_{z\in\mathbb{S}^1}\left\|\til{R}(\e,z)\right\|
\le C,\qquad
\|P_\e-P_0\|
\le C\e,
\een
for any $\e\ge0$ small enough and $z\in\mathbb{S}^1$.
\end{lemma}

\begin{remark}
Note that the matrix $U(\e)|_{A_0}$ is in general non-normal.
Resolvent estimates for a non-normal operator are difficult even if the spectral parameter is away from the spectrum of the operator (see \cite[Formula (1.15)]{DaSi}).
We here use \cite[Theorem 1]{DaSi} to manage this problem.
\end{remark}

\begin{proof}[Proof of Lemma~\ref{lem:proj-estimate}]
The modulus of each eigenvalue of $U(0)|_{A_0}$ is  either one or zero.
This is due to the condition \eqref{eq:free-motion-e=0} which completely separates the quantum walks on the interior and exterior.
More precisely, for an arc $a\in A$ belonging to the exterior, that is, $\ope{supp}\,U(0)^k\dl_a\in A^\out$  and $\ope{supp}\,U(0)^{-k}\dl_a\in A^\inc$ holds for large $k>0$, one has $U(0)|_{A_0}^k\dl_a=0$. 
This means $\dl_a$ belongs to the generalized eigenspace associated with the eigenvalue zero.
The action of $U(0)$ on a function $\pphi\in\cH$ supported in the interior and that of $U(0)|_{A_0}$ on $\chi(A_0)\pphi$ are completely same. This implies the coincidence of the eigenvalues on $\mathbb{S}^1$ of $U(0)$ and $U(0)|_{A_0}$.

Under the assumption of the right-differentiablility of $U(\e)$ at $\e=0$, each eigenvalue of $U(\e)|_{A_0}$ converges to an eigenvalue of $U(0)|_{A_0}$ as $\e\to+0$.
Recall that, for each $\lambda\in \ope{Res}(U(0))\cap\mathbb{S}^1$, the unique eigenvalue $\lambda_\e(\lambda)$ converging to $\lambda$ satisfies $|\lambda_\e-\lambda|=\ord(\e)$ \cite{Kato}.
Then there exists $C>0$ such that $\ope{Res}(U(\e))\setminus E(\e)=\{\lambda_\e(\lambda);\,\lambda\in\ope{Res}(U(0))\cap\mathbb{S}^1\}\subset\{z\in\C;\,|z|\ge1-C\e\}$ holds for $\e\ge0$ small enough.
One then has
\be\label{eq:exp-proj}
I-P_\e=\sum_{\lambda\in\ope{Res}(U(0))\cap\mathbb{S}^1}P_{\lambda_\e(\lambda),\e}.
\ee
Recall that $P_{\lambda_\e(\lambda),\e}$ is the projection onto the eigenspace (since $\lambda_\e(\lambda)$ is simple) associated with the simple eigenvalue $\lambda_\e(\lambda)$ of $U(\e)|_{A_0}$.
Each one of these projection is also right-differentiable at $\e=0$ due to the simplicity of the eigenvalue $\lambda$.
In particular, one has $\left\|P_{\lambda_\e(\lambda),\e}-P_{\lambda,0}\right\|=\ord(\e)$.
This with \eqref{eq:exp-proj} implies the boundedness of $I-P_\e$ and $P_\e$ uniformly for small $\e\ge0$, and
\begin{align*}
\left\|P_\e-P_0\right\|=\left\|(I-P_\e)-(I-P_0)\right\|\le\sum_{\lambda\in\ope{Res}(U(0))\cap\mathbb{S}^1}\left\|P_{\lambda_\e(\lambda),\e}-P_{\lambda,0}\right\|=\ord(\e).
\end{align*}

Concerning the estimate of the reduced resolvent, we remark that one has
\ben
\til{R}(\e,z)=(U(\e)|_{A_0}P_\e-z)^{-1}P_\e.
\een
The above equality can be easily checked by the definition of $\til{R}(\e,z)$ and the following property of $P_\e$ for any $\lambda\in E(\e)$:
\ben
P_\e P_{\lambda,\e}= P_{\lambda,\e}P_\e=P_{\lambda,\e},\qquad
P_\e U(\e)|_{A_0}=U(\e)|_{A_0} P_\e.
\een

According to \cite[Theorem 1]{DaSi}, if an $n\times n$-matrix $M$ and a complex number $z\in\C$ satisfies
\ben
|z|\ge \|M\|,
\een
and if $z$ is not an eigenvalue of $M$, one has
\ben
\|(M-z)^{-1}\|\le\frac{\cot(\pi/4n)}{\ope{dist}(z,\ope{EV}(M))},\qquad
\ope{dist}(z,\ope{EV}(M))=\min_{\lambda\in\ope{EV}(M)}|z-\lambda|,
\een
where $\ope{EV}(M)$ denotes the set of eigenvalues of $M$. 
By applying this formula, we obtain
\ben
\|(U(\e)|_{A_0}P_\e-z)^{-1}\|_{\C^{A_0}\to\C^{A_0}}\le2\cot\left(\frac{\pi}{4\#A_0}\right),
\een
for any $z\in\mathbb{S}^1$ since the modulus of each eigenvalue of $U(\e)|_{A_0}$ is less than or equal to $1/2$.
\end{proof}

\begin{proof}[Proof of Proposition~\ref{prop:res-esti}]
For the reduced resolvents, one has the following resolvent identity:
\be\label{eq:res-eq}
\begin{aligned}
&\til{R}(\e,z)-\til{R}(0,z)
\\&=\til{R}(\e,z)(1-P_0)-(1-P_\e)\til{R}(0,z)
-\til{R}(\e,z)\left(U(\e)|_{A_0}-U(0)|_{A_0}\right)\til{R}(0,z).
\end{aligned}
\ee
Let us prove it. 
By definition, one has
\ben
(U(\e)|_{A_0}-z)\til{R}(\e,z)=\til{R}(\e,z)(U(\e)|_{A_0}-z)=P_\e\qquad(\e\ge0).
\een
Then, one gets $\til{R}(\e,z)(U(\e)|_{A_0}-z)\til{R}(0,z)=P_\e\til{R}(0,z)$,
 and also
\begin{align*}
\til{R}(\e,z)(U(\e)|_{A_0}-z)\til{R}(0,z)
&=\til{R}(\e,z)\left((U(0)|_{A_0}-z)+(U(\e)|_{A_0}-U(0)|_{A_0})\right)\til{R}(0,z)
\\&=\til{R}(\e,z)P_0+\til{R}(\e,z)(U(\e)|_{A_0}-U(0)|_{A_0})\til{R}(0,z).
\end{align*}
By combining the above two expressions of $\til{R}(0,z)(U(\e)|_{A_0}-z)\til{R}(\e,z)$, one obtains
\ben
\til{R}(\e,z)P_0-P_\e\til{R}(0,z)=-\til{R}(\e,z)(U(\e)|_{A_0}-U(0)|_{A_0})\til{R}(0,z).
\een
The resolvent identity \eqref{eq:res-eq} is a consequence of this formula.

Then, let us estimate each term of the right-hand-side of the resolvent identity \eqref{eq:res-eq}.
Since $P_\e$ is a projection, one has $(I-P_\e)\til{R}(\e,z)=\til{R}(\e,z)(I-P_\e)=0$.
This with Lemma~\ref{lem:proj-estimate} implies 
\ben
\|\til{R}(\e,z)(1-P_0)\|=\|\til{R}(\e,z)(P_\e-P_0)\|\le\|\til{R}(\e,z)\|\|P_\e-P_0\|=\ord(\e),
\een
and similarly $\|(1-P_\e)\til{R}(0,z)\|\le\|P_\e-P_0\|\|\til{R}(0,z)\|=\ord(\e)$.
Since $U(\e)$ is  right-differentiable and $\#A_0$ is finite, one has $\|U(\e)|_{A_0}-U(0)|_{A_0}\|=\ord(\e)$ and
\ben
\|\til{R}(\e,z)(U(\e)|_{A_0}-U(0)|_{A_0})\til{R}(0,z)\|\le
\|\til{R}(\e,z)\|\|U(\e)|_{A_0}-U(0)|_{A_0}\|\|\til{R}(0,z)\|=\ord(\e).
\een
This ends the proof.
\end{proof}

\subsection{Proof of Theorem~\ref{thm:non-resonant}}
We here prove our main theorem.
In the first part of the proof below, the formula \eqref{eq:discrepancy-thm} is shown.
With this formula, the other part of the theorem is given by estimates of the matrix $M_{\lambda,\e}$.
\begin{proof}[Proof of Theorem~\ref{thm:non-resonant}]
Recall that the generalized eigenfunction is constructed by using the resolvent of $U|_{A_0}$ by formulae \eqref{eq:local-resolvent} and \eqref{eq:construct-gen-ef}.
Let $\pphi_{\alpha^\inc,\e,z}$ be the generalized eigenfunction for $U(\e)$.
According to the construction there, the outgoing data $\alpha^\out(\e)=\alpha^\out(\alpha^\inc,\e,z)=\chi(\Omega^\out)\pphi_{\alpha^\inc,\e,z}$ is given by
\ben
\alpha^\out(\e)=\chi(\Omega^\out)U(\e)\left(1-\chi(A_0)^*(U(\e)|_{A_0}-z)^{-1}\chi(A_0)U(\e)\right)\chi(\Omega^\inc)^*\alpha^\inc.
\een
As we have seen in the proof of Proposition~\ref{prop:gen-ef}, $\chi(A_0)U(\e)\chi(\Omega^\inc)^*\alpha^\inc$ is orthogonal to any eigenvector associated with an eigenvalue of $U(\e)|_{A_0}$ of modulus one.
This implies that the resolvent can be replaced with the reduced resolvents, that is,
\ben
\alpha^\out(\e)=\chi(\Omega^\out)U(\e)\left(1-\chi(A_0)^*\left(\til{R}(\e,z)+Q(\e,z)\right)\chi(A_0)U(\e)\right)\chi(\Omega^\inc)^*\alpha^\inc.
\een
The difference $\alpha^\out(\e)-\alpha^\out(0)$ is then expressed as
\begin{align*}
\alpha^\out(\e)-\alpha^\out(0)
&=\chi(\Omega^\out)\left[U(s)\left(1-\chi(A_0)^*\til{R}(s,z)\chi(A_0)U(s)\right)\right]_{s=0}^\e\chi(\Omega^\inc)^*\alpha^\inc
\\&\qquad-\chi(\Omega^\out)U(\e)\chi(A_0)^*Q(\e,z)\chi(A_0)U(\e)\chi(\Omega^\inc)^*\alpha^\inc.
\end{align*}
We here use the notation $[F(s)]_{s=0}^\e=F(\e)-F(0)$ and the fact $Q(0,z)=0$.
By the same argument as in the proof of the resonance expansion (Subsection~\ref{sec:prf-res-exp}), the second term of the right-hand-side is computed as
\begin{align*}
&-\chi(\Omega^\out)U(\e)\chi(A_0)^*Q(\e,z)\chi(A_0)U(\e)\chi(\Omega^\inc)^*\alpha^\inc
\\&\qquad=\sum_{\substack{\lambda\in\ope{Res}(U(0))\cap\mathbb{S}^1,\\ |\lambda_\e(\lambda)|\neq1}}
M_{\lambda,\e}(z)\alpha^\inc
=\sum_{\lambda\in\ope{Res}(U(0))\cap\mathbb{S}^1}
M_{\lambda,\e}(z)\alpha^\inc.
\end{align*}
Note that $|\lambda_\e(\lambda)|=1$ implies that $\lambda_\e$ is an eigenvalue of $U(\e)$, and $M_{\lambda,\e}(z)=0$ for any $z\in\C$ since both $\chi(\Omega^\inc)\pphi_{\lambda,\e}^\circledast$ and $\chi(\Omega^\out)\pphi_{\lambda,\e}$ vanish identically (see Lemma~\ref{lem:on-tails}).
For the remaining part, one has
\begin{align*}
&\left[U(s)\left(1-\chi(A_0)^*\til{R}(s,z)\chi(A_0)U(s)\right)\right]_{s=0}^\e
\\&=(U(\e)-U(0))\left(1-\chi(A_0)^*\til{R}(0,z)\chi(A_0)U(0)\right)
+U(\e)\chi(A_0)^*\til{R}(0,z)\chi(A_0)(U(\e)-U(0))
\\&\quad+U(\e)\chi(A_0)^*(\til{R}(\e,z)-\til{R}(0,z))U(\e).
\end{align*}
The norm as $\cH\to\cH$ of the first and second terms are estimated as $\ord(\e)$ by using the unitarity of $U(\e)$, the estimate $\|\til{R}(0,z)\|=\ord(1)$ of Lemma~\ref{lem:proj-estimate}, and $\|U(\e)-U(0)\|_{\cH\to\cH}=\ord(\e)$ (Condition~\ref{cond:pert}).
For the third term, we use the estimate of Proposition~\ref{prop:res-esti}.
This ends the proof of \eqref{eq:discrepancy-thm}.

Let us estimate $M_{\lambda,\e}(z)$.
For any $z\in\mathbb{S}^1$, \eqref{eq:widths} and  \eqref{eq:esti-omega-in} implies
\begin{align*}
\|M_{\lambda,\e}(z)\|_{\C^{\Omega^\inc}\to\C^{\Omega^\out}}
&=
\left\|M_{\lambda,\e}(z)\frac{\chi(\Omega^\inc)\pphi_{\lambda,\e}^\circledast}{\|\chi(\Omega^\inc)\pphi_{\lambda,\e}^\circledast\|}
\right\|_{\C^{\Omega^\out}}
\\&=\left|\frac{\lambda_\e^2}{z-\lambda_\e}\right|\|\chi(\Omega^\inc)\pphi_{\lambda,\e}^\circledast\|\|\chi(\Omega^\out)\pphi_{\lambda,\e}\|
\\&=\left|\frac{1-|\lambda_\e|^2}{z-\lambda_\e}\right|
\|\chi(A_0)\pphi_{\lambda,\e}^\circledast\|\|\chi(A_0)\pphi_{\lambda,\e}\|.
\end{align*}
The normalized eigenvectors $\chi(A_0)\pphi_{\lambda,\e}$ and $\chi(A_0)\pphi_{\lambda,\e}^\circledast$ are given by
\ben
\pphi_{\lambda,\e}=(P_{\lambda_\e,\e}\pphi_{\lambda,0},\pphi_{\lambda,0})^{-1}P_{\lambda_\e,\e}\pphi_{\lambda,0},
\quad
\pphi_{\lambda,\e}^\circledast=(P_{\lambda_\e,\e}^*\pphi_{\lambda,0},\pphi_{\lambda,0})^{-1}P_{\lambda_\e,\e}^*\pphi_{\lambda,0}.
\een
This implies that $\|\chi(A_0)\pphi_{\lambda,\e}^\circledast\|\|\chi(A_0)\pphi_{\lambda,\e}\|=\ord(1)$.
Since $|\lambda_\e-\lambda|=\ord(\e)$, one has $1-|\lambda_\e|^2=\ord(\e)$.
This shows $M_{\lambda,\e}(z)=\ord(\e)$ when $|z-\lambda_\e|\ge c$ for some $c>0$ independent of $\e$.

On the other hand, the Cauchy-Schwarz inequality with the normalization condition \eqref{eq:normalize-res-st} implies $\|\chi(A_0)\pphi_{\lambda,\e}^\circledast\|\|\chi(A_0)\pphi_{\lambda,\e}\|\ge1$.
The estimate \eqref{eq:esti-M-below} follows from this with
\ben
\left(\left|\frac{1-|\lambda_\e|^2}{z-\lambda_\e}\right|\right)\Big|_{z=\lambda_\e/|\lambda_\e|}
=\frac{(1+|\lambda_\e|)(1-|\lambda_\e|)}{1-|\lambda_\e|}
=1+|\lambda_\e|.
\een
\end{proof}

\subsection{Proof of Corollary~\ref{cor:res-tun}}
The proof provided below of Corollary~\ref{cor:res-tun} consists in the following elementary lemma shown by the triangular inequality.
\begin{lemma}\label{lem:triangular}
Let $v_1=v_1(\e)$ and $v_2=v_2(\e)$ be $\C^N$-valued functions of small $\e\ge0$ satisfying
\ben
\|v_j\|=1,\quad\|v_1(\e)-v_2(\e)\|\ge2-C\e,
\een
for some $\e$-independent constant $C>0$.
Then, one has
\ben
v_1(\e)=-v_2(\e)+\ord(\sqrt{\e})=\frac12(v_1(\e)-v_2(\e))+\ord(\sqrt{\e}).
\een
\end{lemma}
\begin{proof}
By the condition $\|v_j\|=1$, one has
\ben
\|v_1-v_2\|^2=\|v_1\|^2+\|v_2\|^2-2\re(v_1,v_2)
=2-2\re(v_1,v_2).
\een
This with the inequality $\|v_1-v_2\|\ge2-C\e$ implies
\ben
1+\re(v_1,v_2)\le 2C\e.
\een
Then one obtains
\ben
\|v_1+v_2\|^2=2(1+\re(v_1,v_2))\le4C\e.
\een
This shows $v_1=-v_2+\ord(\sqrt{\e})$.
The formula $v_1=2^{-1}(v_1-v_2)+\ord(\sqrt{\e})$ follows immediately.
\end{proof}
We also prepare the following lemma on the symmetry of the pair of outgoing and incoming resonant states.

\begin{lemma}\label{eq:sym-out-inc-res-st}
Let $\pphi_{\lambda,\e}$ and $\pphi_{\lambda,\e}^\circledast$ be a pair of outgoing and incoming resonant state associated with a resonance $\lambda_\e(\lambda)$ for $\lambda\in\ope{Res}(U(0))\cap\mathbb{S}^1$ such that $|\lambda_\e(\lambda)|<1$ holds for positive small $\e$. 
For any $J\subset\{1,2,\ldots,N\}$, one has
\ben
\frac{\bigl\|\mathbbm{1}_{\Omega_J^\out}\chi(\Omega^\out)\pphi_{\lambda,\e}\bigr\|}{\left\|\chi(\Omega^\out)\pphi_{\lambda,\e}\right\|}
=\frac{\bigl\|\mathbbm{1}_{\Omega_J^\inc}\chi(\Omega^\inc)\pphi_{\lambda,\e}^\circledast\bigr\|}{\|\chi(\Omega^\inc)\pphi_{\lambda,\e}^\circledast\|}+\ord(\sqrt{\e}).
\een
\end{lemma}
\begin{proof}
Fix $\lambda\in\ope{Res}(U(0))\cap\mathbb{S}^1$. 
Put $\til{\alpha}^\inc=\til{\alpha}^\inc(\lambda,\e):=\|\chi(\Omega^\inc)\pphi_{\lambda,\e}^\circledast\|^{-1}\chi(\Omega^\inc)\pphi_{\lambda,\e}^\circledast$, and $\til{z}=\til{z}(\lambda,\e):=\lambda_\e/|\lambda_\e|$.
This normalized function $\til{\alpha}^\inc \in\C^{\Omega^\inc}$ is the maximizer of the action of the rank one operator $M_{\lambda,\e}(\til{z})$.
The estimate \eqref{eq:esti-M-below} gives
\be\label{eq:esti-below-max}
\left\|M_{\lambda,\e}\left(\til{z}\right)\right\|_{\C^{\Omega^\inc}\to\C^{\Omega^\out}}
=\left\|M_{\lambda,\e}\left(\til{z}\right)\til{\alpha}^\inc\right\|_{\C^{\Omega^\out}}\ge2- C\e.
\ee
We also have an estimate from above
\ben
\|M_{\lambda,\e}(\til{z})\|
=\|\Sigma(\e,\til{z})-\Sigma(0,\til{z})+\ord(\e)\|
\le 2+\ord(\e).
\een
This implies the existence of the constant $c_{J,\e}$ of modulus one such that
\be\label{eq:const-cJ}
M_{\lambda,\e}\left(\til{z}\right)\til{\alpha}^\inc
=2c_{J,\e}\frac{\chi(\Omega^\out)\pphi_{\lambda,\e}}{\|\chi(\Omega^\out)\pphi_{\lambda,\e}\|},
\ee
since the range of $M_{\lambda,\e}$ is spanned by $\chi(\Omega^\out)\pphi_{\lambda,\e}$.

For any $z\in\mathbb{S}^1$ and $\e\ge0$, one has $\|\Sigma(\e,z)\til{\alpha}^\inc\|=1$ since the scattering matrix is isometric.
Then Lemma~\ref{lem:triangular} with \eqref{eq:esti-below-max} and the formula
\ben
\Sigma(\e,\til{z})\til{\alpha}^\inc-\Sigma(0,\til{z})\til{\alpha}^\inc=M_{\lambda,\e}\left(\til{z}\right)\til{\alpha}^\inc+\ord(\e),
\een
gives
\ben
\Sigma(\e,\til{z})\til{\alpha}^\inc
=-\Sigma(0,\til{z})\til{\alpha}^\inc+\ord(\sqrt{\e})=\frac12M_{\lambda,\e}\left(\til{z}\right)\til{\alpha}^\inc+\ord(\sqrt{\e}).
\een
One deduces from this with \eqref{eq:const-cJ} the formula
\be\label{eq:sigma-m}
\Sigma(0,\til{z})\til{\alpha}^\inc
=-\frac12 M_{\lambda,\e}(\til{z})\til{\alpha}^\inc+\ord(\sqrt{\e})
=
-c_{J,\e}\frac{\chi(\Omega^\out)\pphi_{\lambda,\e}}{\|\chi(\Omega^\out)\pphi_{\lambda,\e}\|}+\ord(\sqrt{\e}).
\ee

Put $\beta_J^\inc:=\mathbbm{1}_{\Omega_J^\inc}\til{\alpha}^\inc$ and $\gamma_J^\inc:=\|\beta_J^\inc\|=\|\mathbbm{1}_{\Omega_J^\inc}\chi(\Omega^\inc)\pphi_{\lambda,\e}^\circledast\|/\|\chi(\Omega^\inc)\pphi_{\lambda,\e}^\circledast\|$
for some $J\subset\{1,2,\ldots,N\}$.
Then one has
\ben
(\beta_J^\inc,\chi(\Omega^\inc)\pphi_{\lambda,\e}^\circledast)=
\frac{\|\mathbbm{1}_{\Omega_J^\inc}\chi(\Omega^\inc)\pphi_{\lambda,\e}^\circledast\|^2}{\|\chi(\Omega^\inc)\pphi_{\lambda,\e}^\circledast\|}
=(\gamma_J^\inc)^2\|\chi(\Omega^\inc)\pphi_{\lambda,\e}^\circledast\|,
\een
and consequently
\ben
M_{\lambda,\e}(\til{z})\beta_J^\inc
=(\gamma_J^\inc)^2M_{\lambda,\e}(\til{z})\til{\alpha}^\inc
=-2(\gamma_J^\inc)^2\Sigma(0,\til{z})\til{\alpha}^\inc+\ord(\sqrt{\e}).
\een
Since $\Sigma(0,\til{z})$ is a diagonal matrix, one also has 
\ben
\Sigma(0,\til{z})\beta_J^\inc
=\mathbbm{1}_{\Omega_J^\out}\Sigma(0,\til{z})\til{\alpha}^\inc+\ord(\e).
\een
Then we obtain
\be\label{eq:res-tun-general-gamma}
\Sigma(\e,\til{z})\beta_J^\inc
=(\Sigma(0,\til{z})+M_{\lambda,\e}(\til{z}))\beta_J^\inc+\ord(\e)
=(-2(\gamma_J^\inc)^2+\mathbbm{1}_{\Omega_{J}^\out})\Sigma(0,\til{z})\til{\alpha}^\inc+\ord(\sqrt{\e}).
\ee
The unitarity of $\Sigma(\e,\til{z})$ implies that the norm of this vector is $\gamma_J^\inc$.
Formula~\eqref{eq:sigma-m} gives the equality
\begin{align*}
(\gamma_J^\inc)^2
&=\left\|(-2(\gamma_J^\inc)^2+\mathbbm{1}_{\Omega_J^\out})c_{J,\e}\frac{\chi(\Omega^\out)\pphi_{\lambda,\e}}{\|\chi(\Omega^\out)\pphi_{\lambda,\e}\|}+\ord(\sqrt{\e})\right\|^2
\\&=(1-2(\gamma_J^\inc)^2)^2(\gamma_J^\out)^2+(-2(\gamma_J^\inc)^2)^2(1-(\gamma_J^\out)^2)+\ord(\sqrt{\e}),
\end{align*} 
with $\gamma_J^\out:=\|\mathbbm{1}_{\Omega_J^\out}\chi(\Omega^\out)\pphi_{\lambda,\e}\|/\|\chi(\Omega^\out)\pphi_{\lambda,\e}\|$.
This equality is equivalent to
\ben
(1-4(\gamma_J^\inc)^2)(\gamma_J^\out)^2
=(1-4(\gamma_J^\inc)^2)(\gamma_J^\inc)^2+\ord(\sqrt{\e}).
\een
When $\gamma_J^\inc\neq1/2$, one obtains $\gamma_J^\inc=\gamma_J^\out+\ord(\sqrt{\e})$.
Otherwise, we apply the above argument for $J^c:=\{1,2,\ldots, N\}\setminus J$ $(\gamma_{J^c}^\inc=\sqrt{3}/2)$.
\end{proof}

\begin{proof}[Proof of Corollary~\ref{cor:res-tun}]
\underline{Proof of \eqref{eq:non-resonant-transm}}
Fix $z\in\mathbb{S}^1\setminus\ope{Res}(U(0))$, a non-empty subset $J\subsetneq\{1,2,\ldots,N\}$, and a normalized incoming data $\alpha^\inc\in\C^{\Omega^\inc}$ with $\ope{supp}\,\alpha^\inc\subset\Omega^\inc_J$.
The estimate $\|\Sigma(\e,z)-\Sigma(0,z)\|=\ord(\e)$ implies the approximation
\ben
\Sigma(\e,z)\alpha^\inc=\Sigma(0,z)\alpha^\inc+\ord(\e).
\een
Since $\Sigma(0,z)$ is a diagonal matrix, $\Sigma(0,z)\alpha^\inc=\ord(\e)$ on $\Omega_{J^c}^\out$.

\noindent\underline{Proof of \eqref{eq:resonant-tunneling-explicit}}
Fix $\lambda\in\ope{Res}(U(0))\cap\mathbb{S}^1$. 
We use the same notation as in the proof of Lemma~\ref{eq:sym-out-inc-res-st}.
The incoming data $\alpha_J^\inc\in\C^{\Omega^\inc}$ defined by \eqref{eq:alpha-Js} satisfies
\ben
\alpha_J^\inc=\frac{\|\chi(\Omega^\inc)\pphi_{\lambda,\e}^\circledast\|}{\|\mathbbm{1}_{\Omega^\inc_J}\chi(\Omega^\inc)\pphi_{\lambda,\e}^\circledast\|}\beta_J,
\een
where $\beta_J$ is defined right after \eqref{eq:sigma-m}.
Under the assumption \eqref{eq:cond-res-tun} (see also Remark~\ref{rem:equiv-cond-res-tun}), one has
\ben
\gamma_J^\inc=\|\beta_J^\inc\|=\frac1{\sqrt{2}}.
\een
The formula \eqref{eq:res-tun-general-gamma} is rewritten as
\ben
\Sigma(\e,\til{z})\alpha_J^\inc=-\sqrt{2}\mathbbm{1}_{\Omega_{J^c}^\out}\Sigma(0,\til{z})\til{\alpha}^\inc+\ord(\sqrt{\e}).
\een
Formula~\eqref{eq:sigma-m} and the unitarity of the scattering matrix $\Sigma(\e,\til{z})$ shows that the right-hand-side is asymptotic to $c_{J,\e}\alpha_{J^c}^\out$.

\noindent\underline{Proof of \eqref{eq:width-of-peak}}
The width of the resonant peak is computed by using
\ben
M_{\lambda,\e}(e^{i\theta}\til{z})\alpha_J^\inc
=\frac{\til{z}-\lambda_\e}{e^{i\theta}\til{z}-\lambda_\e}M_{\lambda,\e}(\til{z})\alpha_J^\inc,
\een
and $\Sigma(0,e^{i\theta}\til{z})=\Sigma(0,\til{z})+\ord(|\theta|)$ $(|\theta|\ll1).$
Then the condition to have the half-height, that is, $T(J,\alpha_J^\inc,\e,e^{i\theta}\til{z})=1/2$, is approximated by
\begin{align*}
(1-|\lambda_\e|)^2
=\frac{|e^{i\theta}-|\lambda_\e||^2}{2}+\ord(\e)
=\frac12\left((1-|\lambda_\e|)^2+|\theta|^2\right)(1+\ord(|\theta|^2))+\ord(\e).
\end{align*}
This with the estimate $1-|\lambda_\e|\le|\lambda-\lambda_\e|=\ord(\e)$ gives the result.
\end{proof}

\subsection{Proof of Theorem~\ref{thm:comfortability}}
Assume for simplicity that $\alpha^\inc$ is normalized: $\|\alpha^\inc\|=1$.
As in the proof of Theorem~\ref{thm:non-resonant}, let $\pphi_{\alpha^\inc,\e,z}$ be the generalized eigenfunction for $U(\e)$ defined by  formulae \eqref{eq:local-resolvent} and \eqref{eq:construct-gen-ef}.
According to the definition there, one has
\ben
\chi(A_0)\pphi_{\alpha^\inc,\e,z}=u(\alpha^\inc,\e,z)=-(\til{R}(\e,z)+Q(\e,z))\chi(A_0)U(\e)\chi(\Omega^\inc)\alpha^\inc.
\een
Then we diduce
\begin{align*}
&\chi(A_0)\pphi_{\alpha^\inc,\e,z}-\chi(A_0)\pphi_{\alpha^\inc,0,z}
\\&=\chi(A_0)\left((\til{R}(0,z)-\til{R}(\e,z))\chi(A_0)U(\e)+\til{R}(0,z)\chi(A_0)(U(0)-U(\e))\right)\chi(\Omega^\inc)\alpha^\inc
\\&\qquad-\chi(A_0)Q(\e,z)\chi(A_0)U(\e)\chi(\Omega^\inc)\alpha^\inc,
\end{align*}
from the orthogonality of $\chi(A_0)U(\e)\chi(\Omega^\inc)\alpha^\inc$ to the eigenspace associated with each eigenvalue of $U(\e)|_{A_0}$ of modulus one (Proposition~\ref{prop:gen-ef}).
The first term is estimated as $\ord(\e)$ by using the estimates proven in Proposition~\ref{prop:res-esti} and Lemma~\ref{lem:proj-estimate}.
The second term is rewritten as
\ben
-\chi(A_0)Q(\e,z)\chi(A_0)U(\e)\chi(\Omega^\inc)\alpha^\inc
=\sum_{\lambda\in\ope{Res}(U(0))\cap\mathbb{S}^1}\frac1{z-\lambda_\e(\lambda)}\chi(A_0)P_{\lambda_\e(\lambda),\e}\chi(A_0)U(\e)\chi(\Omega^\inc)\alpha^\inc.
\een
As shown in \eqref{eq:proj-exp}, one has $P_{\lambda_\e(\lambda),\e}f=(f,\chi(A_0)\pphi_{\lambda,\e}^\circledast)\pphi_{\lambda,\e}$.
Then for each $\lambda\in\ope{Res}(U(0))\cap\mathbb{S}^1$ with $|\lambda_\e|<1$ for small $\e>0$, one obtains
\begin{align*}
\chi(A_0)P_{\lambda_\e(\lambda),\e}\chi(A_0)U(\e)\chi(\Omega^\inc)^*\alpha^\inc
&=(\alpha^\inc,\chi(\Omega^\inc)U(\e)^*\chi(A_0)\pphi_{\lambda,\e}^\circledast)\chi(A_0)\pphi_{\lambda,\e}
\\&=\lambda_\e(\alpha^\inc,\chi(\Omega^\inc)\pphi_{\lambda,\e}^\circledast)\chi(A_0)\pphi_{\lambda,\e}.
\end{align*}
As a conclusion, we have obtained
\be\label{eq:disc-inside}
\chi(A_0)\pphi_{\alpha^\inc,\e,z}-\chi(A_0)\pphi_{\alpha^\inc,0,z}
=\sum_{\lambda\in\ope{Res}(U(0))\cap\mathbb{S}^1}\frac{\lambda_\e(\lambda)(\alpha^\inc,\chi(\Omega^\inc)\pphi_{\lambda,\e}^\circledast)}{z-\lambda_\e(\lambda)}\chi(A_0)\pphi_{\lambda,\e}+\ord(\e).
\ee

Let us estimate each term of the right-hand-side of \eqref{eq:disc-inside}.
The Cauchy-Schwarz inequality with \eqref{eq:esti-omega-in} implies
\be\label{eq:non-res-esti-comf}
\left|(\alpha^\inc,\chi(\Omega^\inc)\pphi_{\lambda,\e}^\circledast)\right|
\le\|\alpha^\inc\|\|\chi(\Omega^\inc)\pphi_{\lambda,\e}^\circledast\|
=\sqrt{|\lambda_\e|^{-2}-1}\,\|\chi(A_0)\pphi_{\lambda,\e}^\circledast\|=\ord(\sqrt{\e}).
\ee
Then for any fixed $z\in\mathbb{S}^1\setminus\ope{Res}(U(0))$, one obtains
\ben
\left\|\chi(A_0)\pphi_{\alpha^\inc,\e,z}-\chi(A_0)\pphi_{\alpha^\inc,0,z}\right\|=\ord(\sqrt{\e}).
\een
This with the definition of the comfortability implies
\ben
\sqrt{\mc{E}(U(\e),\alpha^\inc,z)}=\left\|\chi(A_0)\pphi_{\alpha^\inc,\e,z}\right\|
\le \left\|\chi(A_0)\pphi_{\alpha^\inc,0,z}\right\|+\ord(\sqrt{\e}),
\een
and
\begin{align*}
\left|\sqrt{\mc{E}(U(\e),\alpha^\inc,z)}-\sqrt{\mc{E}(U(0),\alpha^\inc,z)}\right|
&=\left|\left\|\chi(A_0)\pphi_{\alpha^\inc,\e,z}\right\|-\left\|\chi(A_0)\pphi_{\alpha^\inc,0,z}\right\|\right|
\\&\le\left\|\chi(A_0)\pphi_{\alpha^\inc,\e,z}-\chi(A_0)\pphi_{\alpha^\inc,0,z}\right\|=\ord(\sqrt{\e}).
\end{align*}
By combining the above two estimates, we obtain
\begin{align*}
&\left|\mc{E}(U(\e),\alpha^\inc,z)-\mc{E}(U(0),\alpha^\inc,z)\right|
\\&=\left(\sqrt{\mc{E}(U(\e),\alpha^\inc,z)}+\sqrt{\mc{E}(U(0),\alpha^\inc,z)}\right)
\left|\sqrt{\mc{E}(U(\e),\alpha^\inc,z)}-\sqrt{\mc{E}(U(0),\alpha^\inc,z)}\right|
=\ord(\sqrt{\e}).
\end{align*}
This ends the proof of \eqref{eq:non-res-comf}.

Fix $\lambda\in\ope{Res}(U(0))\cap\mathbb{S}^1$ such that $|\lambda_\e(\lambda)|<1$ for positive small $\e>0$.
Then for $\til{\alpha}^\inc=\|\chi(\Omega^\inc)\pphi_{\lambda,\e}^\circledast\|^{-1}\chi(\Omega^\inc)\pphi_{\lambda,\e}^\circledast$, one has
\ben
(\til{\alpha}^\inc,\chi(\Omega^\inc)\pphi_{\lambda,\e})=\|\chi(\Omega^\inc)\pphi_{\lambda,\e}^\circledast\|.
\een
This with \eqref{eq:esti-omega-in} shows the estimate
\ben
\left\|\frac{\lambda_\e(\lambda)(\til{\alpha}^\inc,\chi(\Omega^\inc)\pphi_{\lambda,\e})}{\til{z}-\lambda_\e(\lambda)}\chi(A_0)\pphi_{\lambda,\e}\right\|
=\left|\lambda_\e^2\frac{1+|\lambda_\e|}{1-|\lambda_\e|}\right|^{1/2}\left\|\chi(A_0)\pphi_{\lambda,\e}^\circledast\right\|\left\|\chi(A_0)\pphi_{\lambda,\e}\right\|,
\een
where $\til{z}=\lambda_\e/|\lambda_\e|\in\mathbb{S}^1$.
By combining this identity with \eqref{eq:disc-inside} and \eqref{eq:non-res-esti-comf}, we obtain
\ben
\left\|\chi(A_0)\pphi_{\alpha^\inc,\e,z}-\chi(A_0)\pphi_{\alpha^\inc,0,z}\right\|
=\left|\lambda_\e^2\frac{1+|\lambda_\e|}{1-|\lambda_\e|}\right|^{1/2}\left\|\chi(A_0)\pphi_{\lambda,\e}^\circledast\right\|\left\|\chi(A_0)\pphi_{\lambda,\e}\right\|+\ord(\sqrt{\e}).
\een
The result follows from the estimate
\ben
\sqrt{\mc{E}(U(\e),\til{\alpha}^\inc,\til{z})}\ge\left|\left\|\chi(A_0)\pphi_{\alpha^\inc,\e,z}-\chi(A_0)\pphi_{\alpha^\inc,0,z}\right\|
-\sqrt{\mc{E}(U(0),\til{\alpha}^\inc,\til{z})}\right|,
\een
with the uniform estimate \eqref{eq:non-pert-comf} with respect to $\alpha^\inc$ with $\|\alpha^\inc\|=1$ and $z\in\mathbb{S}^1$ of $\mc{E}(U(0),\alpha^\inc,z)$.

\section*{Acknowledgements}
The author is supported by Grant-in-Aid for JSPS Fellows (Grant No. JP22KJ2364).
The author would like to express sincere gratitude to Kouichi Taira for many fruitful discussions and to Hisashi Morioka, Etsuo Segawa, and Ryuta Ishikawa for inspiring work that motivated this research.

\appendix
\section{Computations for quantum walks on the line}\label{app:on-the-line}
We here show the computation for the results written in Section~\ref{sec:line}.
Recall (e.g., \cite[Section 5]{HMS}) that the identity $(\til{U}-z)\pphi=0$ holds if and only if 
\ben
\begin{pmatrix}(1\ 0)\pphi(x)\\(0\ 1)\pphi(x+1)\end{pmatrix}
=\mc{T}_{z}(x)\begin{pmatrix}(1\ 0)\pphi(x-1)\\(0\ 1)\pphi(x)\end{pmatrix}
\een
holds for every $x\in\Z$, where the matrix $\mc{T}_{z}(x)$ called \textit{transfer matrix} is defined by
\ben
\mc{T}_z(x)=\frac1{C(x)_{11}}\begin{pmatrix}z&-C(x)_{12}\\C(x)_{21}&z^{-1}\det C(x)\end{pmatrix}.
\een
Here, $C(x)_{jk}$ stands for the $(j,k)$-entry of $C(x)$.
Then the generalized eigenfunction $\pphi_z$ satisfying \eqref{eq:gen-eigen-tildeU} is computed by
\ben
\begin{pmatrix}(1\ 0)\pphi(x)\\(0\ 1)\pphi(x+1)\end{pmatrix}
=\mc{T}_{z}(x)\mc{T}_z(x-1)\cdots \mc{T}_z(0)\begin{pmatrix}(1\ 0)\pphi(-1)\\(0\ 1)\pphi(0)\end{pmatrix}
=\frac1z \mc{T}_{z}(x)\mc{T}_z(x-1)\cdots \mc{T}_z(0)\begin{pmatrix}1\\0\end{pmatrix},
\een
for $x\ge0$.
Put
\ben
\mc{T}_z(x,0):=\frac1z \mc{T}_{z}(x)\mc{T}_z(x-1)\cdots \mc{T}_z(0).
\een

In the double barrier problem, that is, under the condition $C(x)=I_2$ $(x\in\Z\setminus\{0,x_0\})$, one obtains 
\begin{align*}
&zC(0)_{11}C(x_0)_{11}\mc{T}_z(x_0,0)
\\&=
\begin{pmatrix}
z(z^{x_0}-C(x_0)_{12}C(0)_{21}z^{-x_0})&-C(0)_{12}z^{x_0}-C(x_0)_{12}(\det C(0))z^{-x_0}\\
C(x_0)_{21}z^{x_0}+C(0)_{21}(\det C(x_0))z^{-x_0}&(\det C(x_0)C(0))z^{-x_0-1}-C(0)_{12}C(x_0)_{21}z^{x_0-1}
\end{pmatrix},
\end{align*}
by a straightforward computation.
The transmission and reflection probabilities $T(z)$ and $R(z)=1-T(z)$ are given by
\ben
T(z)=\left|\frac1{\mathtt{a}(z)}\right|^2,\qquad
R(z)=\left|\frac{\mathtt{b}(z)}{\mathtt{a}(z)}\right|^2,
\een
where
\ben
\begin{pmatrix}\mathtt{a}(z)\\\mathtt{b}(z)\end{pmatrix}=
\begin{pmatrix}(1\ 0)\pphi(x_0)\\(0\ 1)\pphi(x_0+1)\end{pmatrix}
=\mc{T}(x_0,0)\begin{pmatrix}1\\0\end{pmatrix}.
\een
This with $z\in\mathbb{S}^1$ proves \eqref{eq:transm-DB}.
An equivalent condition for $z\in\C\setminus\{0\}$ to be a resonance is $\mathtt{a}(z)=0$.
This gives \eqref{eq:QC-DB}.

Let us prove \eqref{eq:arg-identity}.
It suffices to show 
\ben
-\frac{\det C(x_0)}{C(x_0)_{12}C(x_0)_{21}}=1-\frac{C(x_0)_{11}C(x_0)_{22}}{C(x_0)_{12}C(x_0)_{21}}>0
\een
We deduce
\ben
-\frac{C(x_0)_{11}C(x_0)_{22}}{C(x_0)_{12}C(x_0)_{21}}=-\frac{C(x_0)_{11}C(x_0)_{22}\overline{C(x_0)_{12}C(x_0)_{21}}}{|C(x_0)_{12}C(x_0)_{21}|^2}
=\left|\frac{C(x_0)_{11}C(x_0)_{21}}{C(x_0)_{12}C(x_0)_{21}}\right|^2\ge0
\een
from the mutual orthogonality of row vectors of the unitary matrix $C(x_0)$: $C(x_0)_{11}\overline{C(x_0)_{21}}=-C(x_0)_{12}\overline{C(x_0)_{22}}$.

In the case of triple barrier problem, we have
\begin{align*}
\mc{T}(x_1,0)=\frac1{zC(0)_{11}C(x_0)_{11}C(x_1)_{11}}\begin{pmatrix}\mathtt{a}(z)&*\\\mathtt{b}(z)&*\end{pmatrix},
\end{align*}
with $\mathtt{a}(z)$ and $\mathtt{b}(z)$ used in \eqref{eq:transm-TB}:
\begin{align*}
\mathtt{a}(z)
&=z^{x_1+1}-C(x_0)_{21}C(x_1)_{12}z^{2x_0-x_1+1}
\\&\qquad -C(0)_{21}C(x_0)_{12}z^{x_1-2x_0+1}-C(0)_{21}C(x_1)_{12}\det(C(x_0))z^{-x_1+1}
\\
\mathtt{b}(z)
&=
C(x_1)_{21}z^{x_1}+C(x_0)_{21}(\det C(x_1))z^{2x_0-x_1}-C(0)_{21}C(x_0)_{12}C(x_1)_{21}z^{x_1-2x_0}
\\&\qquad+C(0)_{21}(\det C(x_0)C(x_1))z^{-x_1}.
\end{align*}

\section{Elementary computations}
\label{app:elementary}
We prove the identity \eqref{eq:elementary-pw}:
\ben
\sum_{k=0}^{N-1}\frac{\mu^{kp}}{z-c\mu^k}
=\frac{Nc^{N-p}z^{p-1}}{z^N-c^N},
\qquad(c>0,\ p=1,2,\ldots,N)
\een
where $\mu=\exp(2i\pi/N)$.

\begin{proof}
Note that one has $\prod_{k=0}^{N-1}(z-c\mu^k)=z^N-c^N$, and consequently
\ben
\sum_{k=0}^{N-1}\frac{\mu^{kp}}{z-c\mu^k}
=\frac1{z^N-c^N}\sum_{k=0}^{N-1}\mu^{kp}\prod_{\substack{j\neq k\\0\le j\le N}}(z-c\mu^j).
\een
The product of this formula is expanded as
\ben
\prod_{\substack{j\neq k\\0\le j\le N}}(z-c\mu^j)
=\sum_{s=0}^{N-1}z^{N-s-1}(-c\mu^k)^s\left(\prod_{1\le l_1<l_2<\cdots<l_s\le N-1}\mu^{\sum_{q=1}^sl_q}\right).
\een
Note that the last factor is independent of $k$, and one obtains
\ben
\sum_{k=0}^{N-1}\frac{\mu^{kp}}{z-c\mu^k}
=\frac1{z^N-c^N}
\sum_{s=0}^{N-1}(-c)^sz^{N-s-1}\left(\prod_{1\le l_1<l_2<\cdots<l_s\le N-1}\mu^{\sum_{q=1}^sl_q}\right)
\sum_{k=0}^{N-1}\mu^{k(s+p)}.
\een
Recall that one has $\sum_{k=0}^{N-1}\mu^{k(s+p)}=N\til{\dl}_{N,s+p}$ since $\mu^{k(s+p)}$ is a root of the equation $z^N-1=(z-1)(\sum_{k=0}^{N-1}z^k)=0$.
This ends the proof. 
\end{proof}

\bibliographystyle{plain}
\bibliography{QW}

\end{document}